\documentclass[journal]{IEEEtran}
\pdfoutput=1
\usepackage{amssymb}
\usepackage{amsfonts}
\usepackage{bbm}
\usepackage{amsfonts}
\usepackage[dvips]{graphicx}
\usepackage[dvips]{color}
\usepackage{graphicx}
\usepackage{amsmath}
\usepackage{fancyhdr}
\usepackage{stfloats}
\usepackage{multirow}
\usepackage{algorithm}
\usepackage{algorithmic}
\usepackage{cite}
\usepackage[bookmarks=false]{}
\usepackage{amsthm}
\usepackage{algorithm}
\usepackage{algorithmic}
\usepackage{mathbbol}
\usepackage{amsfonts}
\usepackage{graphicx}
\usepackage{amsmath}
\usepackage{amssymb}
\usepackage{latexsym}
\usepackage{graphicx}
\usepackage{cite}
\usepackage{url}
\usepackage{stfloats}
\usepackage{mathrsfs}
\usepackage{setspace}
\usepackage{booktabs}
\usepackage{latexsym}
\usepackage{url}
\usepackage{stfloats}
\usepackage{mathrsfs}
\theoremstyle{plain}
\newtheorem{definition}{Definition}
\newtheorem{theorem}{Theorem}

\usepackage{cite}
\usepackage{setspace}
\usepackage{textcomp}
\usepackage{makecell}
\usepackage{etoolbox}
\usepackage{diagbox}
\usepackage{caption}
\usepackage{float}
\newtheorem*{remark}{Remark}
\usepackage[table]{xcolor}
\usepackage{multirow}
\usepackage[tight,footnotesize]{subfigure}
\usepackage[]{hyperref}
\hypersetup{
	bookmarksnumbered=true,     
	bookmarksopen=true,         
	bookmarksopenlevel=1,       
	colorlinks=true,
	linkcolor=black            
}
\begin{document}
\clearpage
\title{\huge Adaptable Semantic Compression and Resource Allocation for Task-Oriented Communications}
%

\author{Chuanhong Liu, Caili Guo, \emph{Senior Member}, \emph{IEEE}, Yang Yang and Nan Jiang
\thanks{This work was supported by the Fundamental Research Funds for the Central Universities (No.2021XD-A01-1).}
\thanks{Chuanhong Liu and Caili Guo are with the Beijing Key Laboratory of Network System Architecture and Convergence, School of Information and Communication Engineering, Beijing University of Posts and Telecommunications, Beijing 100876, China (e-mail:2016\_liuchuanhong@bupt.edu.cn; guocaili@bupt.edu.cn)}
\thanks{Yang Yang is with the Beijing Laboratory of Advanced Information Networks, School of Information and Communication Engineering, Beijing University of Posts and Telecommunications, Beijing 100876, China (e-mail:yangyang01@bupt.edu.cn)}
\thanks{Nan Jiang is with the State Key Laboratory of Networking and Switching Technology, School of Information and Communication Engineering, Beijing University of Posts and Telecommunications, Beijing 100876, China (e-mail:nan.jiang@bupt.edu.cn)}
}

%
%
%
\maketitle
\pagestyle{headings}
\vspace{0cm}
\begin{abstract}
Task-oriented communication is a new paradigm that aims at providing efficient connectivity for accomplishing intelligent tasks rather than the reception of every transmitted bit. In this paper, a deep learning-based task-oriented communication architecture is proposed where the user extracts, compresses and transmits semantics in an end-to-end (E2E) manner. Furthermore, an approach is proposed to compress the semantics according to their importance relevant to the task, namely, adaptable semantic compression (ASC). Assuming a delay-intolerant system, supporting multiple users indicates a problem that executing with the higher compression ratio requires fewer channel resources but leads to the distortion of semantics, while executing with the lower compression ratio requires more channel resources and thus may lead to a transmission failure due to delay constraint. To solve the problem, both compression ratio and resource allocation are optimized for the task-oriented communication system to maximize the success probability of tasks. Specifically, due to the nonconvexity of the problem, we propose a compression ratio and resource allocation (CRRA) algorithm by separating the problem into two subproblems and solving iteratively to obtain the convergent solution. Furthermore, considering the scenarios where users have various service levels, a compression ratio, resource allocation, and user selection (CRRAUS) algorithm is proposed to deal with the problem. In CRRAUS, users are adaptively selected to complete the corresponding intelligent tasks based on branch and bound method at the expense of higher algorithm complexity compared with CRRA. Simulation results show that the proposed ASC approach can reduce the size of transmitted data by up to 80\% without reducing task success, and the proposed CRRA and CRRAUS algorithms can obtain at least 15\% and 10\% success gains over baseline algorithms, respectively.
\end{abstract}


\vspace{-0.cm}
\section{Introduction}
\label{sec:intro}


\IEEEPARstart{T}{o support} the rapid development of artificial intelligence (AI), providing connectivity for intelligent tasks performed on the edge is one of the key applications in future wireless communication systems \cite{5G,Letaief_Roadmap}. These tasks are deemed a machine understanding and performing tasks automatically in a fashion close to human cognition, such as recognizing specific content in a text or image. To provide connectivity for such tasks, the goal of communication is no longer the accurate reception of every transmitted bit but to transmit the meaningful content of raw data to accomplish the tasks. This communication paradigm refers to ``task-oriented" communication, which has attracted extensive attention from industry and academia\cite{Nine,Hoydis} and has been identified as one of the core challenges in the next-generation wireless communication systems \cite{Walid_6G}. Some recent researches show that task-oriented communication is transmitting the semantics of the source information with respect to the requirements of tasks, and thus it has great potential to reduce the network traffic and thus alleviate spectrum shortage\cite{Qin_survey}. 

In the existing researches, semantic information is defined as the meaning underlying the raw data\cite{Wang_Attention}, which is always abstract and subjective. Due to the subjective nature of semantic information, the same data may have different semantics in various intelligent tasks, and thus semantics is always highly related to the task. Correspondingly, semantic compression also depends on the intelligent task, which is challenging due to the lack of a unified compression criterion. In addition, employing semantic communications in wireless networks faces several challenges, including the semantic theory, semantic extraction method, semantic-oriented resource allocation, and the performance metrics, which motivates us to investigate more in this area.

Some prior studies have been dedicated to designing fundamental frameworks for semantic communication from the informative-theoretical perspective\cite{framework_Kalfa,framework_Uysal,framework_Kountouris,framework_Seo,framework_Lan,framework_Yang,framework_Shi,framework_Zhang,Guler_game}. In \cite{framework_Kalfa}, the authors discussed semantic transformations of different sources for popular tasks in the field and presented the semantic communication system design for different types of sources. An envisioned end-to-end (E2E) architecture of semantic communication framework was proposed in \cite{framework_Kountouris}, which introduces the semantic sampling that allows each smart device to control its traffic via semantic-aware active sampling. The authors in \cite{framework_Lan} classified semantic communications into human-to-human (Level 2), human-to-machine (Level 2 and Level 3), and machine-to-machine (Level 3) communications. In \cite{framework_Yang}, the framework of task-oriented semantic communication was proposed. The authors in \cite{Guler_game} proposed an E2E learning-driven architecture of semantic communication to integrate the semantic inference and physical layer communication problems, where the transceiver is optimized jointly to reach Nash equilibrium while minimizing the average semantic errors. Recently, deep learning (DL) has emerged as a popular solution for semantic communications due to its powerful feature extraction capability \cite{Farsad,Xie_Deep,Xie_lite,Gunduz_JSCC,Lee,retrival1,retrival2,Liu_AIoT,MU-DeepSC1,MU-DeepSC2}.
Farsad \emph{et al.}\cite{Farsad} developed a long short-term memory (LSTM) enabled joint source-channel coding (JSCC) for the transmission of text data. It shows the great potential of DL-enabled JSCC compared to the conventional communication system. The authors in \cite{Xie_Deep} proposed a semantic communication system based on Transformer, which clarified the concept of semantic information at the sentence level. Based on \cite{Xie_Deep}, the authors in \cite{Xie_lite} further proposed a lite distributed semantic communication system, making the model easier to deploy on the Internet of things (IoT) devices. For image data transmission, the authors in \cite{Gunduz_JSCC} presented a JSCC scheme based on convolutional neural networks (CNN) to transmit image data over wireless channel, which can jointly optimize various modules of the communication system. 

More recently, DL-driven communication architecture considering the semantics of specific tasks has been proposed\cite{Lee,retrival1,retrival2,Liu_AIoT,MU-DeepSC1,MU-DeepSC2,Shao_IB}. Lee \emph{et al.}\cite{Lee} designed a joint transmission-classification system for images, in which the receiver outputs image classification results directly. It has been verified that such a joint design achieved higher classification accuracy than performing image recovery and classification separately. Jankowski \emph{et al.}\cite{retrival1,retrival2} considered image-based re-identification for persons or cars as the communication task, where two schemes were proposed to improve the retrieval accuracy. In our prior work \cite{Liu_AIoT}, an intelligent task-oriented communication method has been proposed for AI of Things (AIoT), in which semantics can be further compressed without performance penalty. For multimodal data transmission, Xie \emph{et al.} \cite{MU-DeepSC1} developed MU-DeepSC for the visual question answering task, where one user transmits text-based questions about images, and the inquiry images are transmitted from another user. Based on MU-DeepSC, a Transformer based framework\cite{MU-DeepSC2} has been developed as a unique structure for serving different tasks. Various tasks have been tested in \cite{MU-DeepSC2} to show its superiority. In summary, the existing works on semantic communications are focused on the implementation of semantic communication systems, in which extracted semantics are compressed via a fixed neural network or directly transmitted without further compression.

Semantic compression aims to lessen the subsequent computing overhead and reduce the amount of transmitted data, consequently reducing the communication burden. Even though the reception data can be reduced by semantic compression, the performance of semantic communication is still restricted due to the neglected management of limited wireless resources. Therefore, it is necessary to study resource allocation policy further to improve semantic communication performance. Furthermore, since users may require transmission service for various intelligent tasks \cite{SLA}, appropriate resource allocation in a semantic aware manner is crucial that guarantees the transmission with prioritized reliabilities. To optimize the performance of semantic communications, the following major issues remain to be solved: 1) \emph{How to adaptively compress semantics with respect to the intelligent tasks?} 2) \emph{How to appropriately allocate communication resources (including bandwidth and transmit power) for compressed data?}


In this paper, we investigate the performance optimization for task-oriented multi-user semantic communication systems. Moreover, two fundamental problems, semantic compression and resource allocation, are studied and solved to improve the performance of semantic communications. To our best knowledge, this is the first work that proposes a theoretical model of semantic compression and resource allocation for task-oriented multi-user semantic communications. The main contributions of this paper are summarized as follows:

\begin{itemize}
	\item[$\bullet$] We design a novel framework for task-oriented multi-user semantic communications that enables users to extract, compress, and transmit the semantics of the raw data effectively to the edge server. The edge server then executes the intelligent task and returns results to users based on the received semantics.
	\item[$\bullet$] An adaptable semantic compression (ASC) approach is proposed to compress extracted semantics based on semantic importance to reduce the communication burden. To complete the ASC, we propose a gradient-based semantic importance evaluation method. The mathematical relationship between the performance of intelligent tasks and semantic compression ratios is then investigated. 
	\item[$\bullet$] Due to wireless resource limitations, users must adaptively determine the optimal semantic compression ratios, and wireless resources must be appropriately allocated to satisfy the transmission delay constraint. The problem is formulated as an optimization problem whose goal is to maximize the success probability of tasks in terms of resource allocation, user selection, and semantic compression ratios. To solve this nonconvex problem, a compression ratio and resource allocation (CRRA) algorithm is proposed for scenarios where users have the same service levels, in which the problem is separated into two subproblems and solved iteratively. 
	\item[$\bullet$] Considering that users have various service levels, we further propose a CRRA with dynamic user selection (CRRAUS) algorithm, in which compression ratio, user selection, and resource allocation are simultaneously optimized. Specifically, the branch and bound method is used to select users based on resources and service levels adaptively.
\end{itemize}
	

The remainder of this paper is organized as follows. The system model and problem formulation are described in Section \ref{sec:system}. Section \ref{sec:details} details the E2E semantic communication and ASC approach. The proposed CRRA algorithm and CRRAUS algorithm are presented in Sections \ref{sec:CRRA} and \ref{sec:CRRAUS}, respectively. Simulation and numerical results are analyzed in Section \ref{sec:simulation}. Section \ref{sec:conclusion} draws some important conclusions.


\vspace{-0.2cm}
\section{System model and Problem formulation}
\label{sec:system}
%
%
We consider a multi-user semantic communication system composed of an edge server and a set ${{\cal U}}$ of $U$ users as illustrated in Fig. \ref{fig:model}. The user aims at gathering data locally and performing an inference task with the assistance of the edge server. To do so, the semantics of raw data is extracted and compressed locally by users before uploading. Then, the semantics is transmitted to the edge server in a scheduled manner, where the edge server allocates channel resource according to channel state information as well as the prior knowledge of semantic compression. Finally, the edge server performs intelligent computing according to the received semantics and returns the result of tasks to users. The users and edge server are equipped with a certain knowledge base to facilitate semantic extraction and compression, where the knowledge base could be different for various applications. In the following, we first introduce the architecture of the end-to-end communication model between a user and the server. Then, we formulate an optimization problem by allocating resources to maximize the performance of semantic communications.

\begin{figure}[t]
	\begin{center}
		\includegraphics[width=1\linewidth]{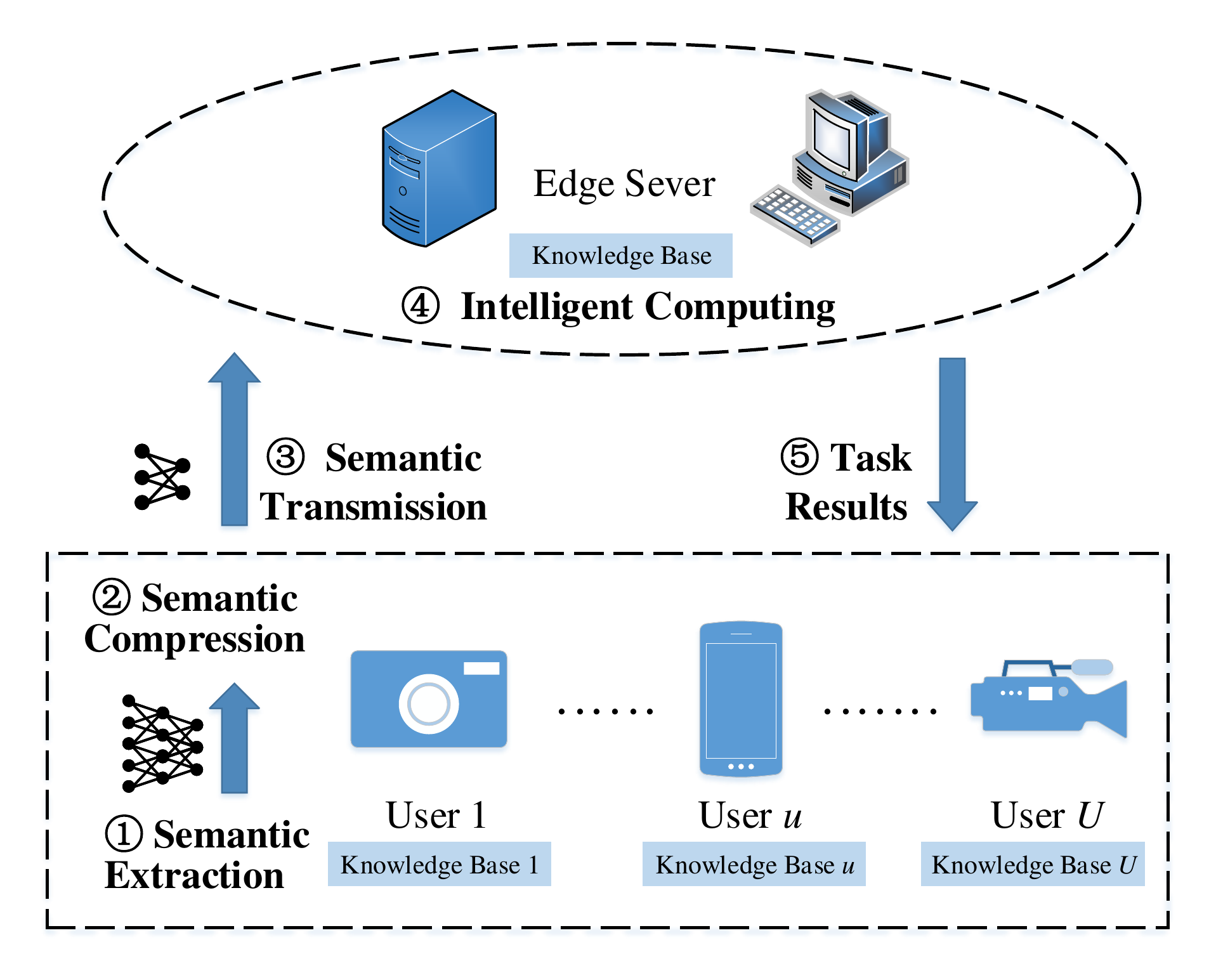}
	\end{center}
	\caption{The structure of task-oriented multi-user semantic communication network.}
	\label{fig:model}
\end{figure}

\subsection{E2E Task-Oriented Semantic Communication System}
\begin{figure*}[htbp]
	\begin{center}
		\includegraphics[width=1\linewidth]{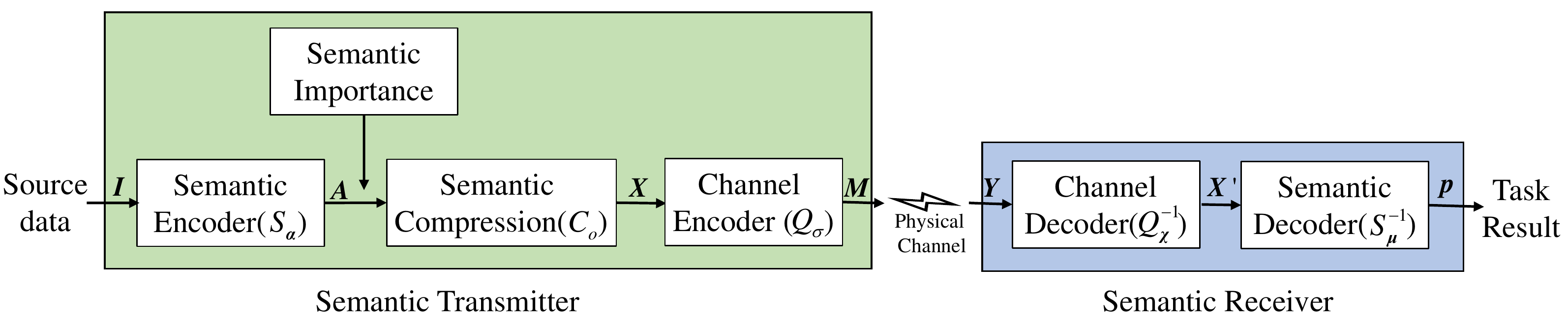}
	\end{center}
	\caption{The framework of proposed task-oriented semantic communication system.}
	\label{fig:structure}
\end{figure*}

We consider an E2E semantic communication system constructed by a neural network architecture as shown in Fig. \ref{fig:structure}. Specifically, the transmitter consists of a semantic encoder to extract the semantic features from the source data, a semantic compression model to compress the semantics to reduce the amount of the transmitted data based on semantic importance, and a channel encoder to generate symbols to facilitate the transmission subsequently. The receiver is composited with a channel decoder for symbol detection and a semantic decoder with the output of semantic concepts with respect to the tasks.

At semantic transmitter, neural networks are first utilized to extract the semantic information from source data ${\boldsymbol I}$, which can be denoted by
\begin{eqnarray}\label{signalform}
	{\boldsymbol {A}} = {S_{\boldsymbol {\alpha }}}({\boldsymbol {I}}),
\end{eqnarray}
where ${S_{\boldsymbol{\alpha }}}(\cdot)$ denotes the semantic encoder network with parameter set ${\boldsymbol{\alpha }}$. The extracted semantics can be a series of semantic features.

Then, the semantic features are compressed by
\begin{eqnarray}\label{}
	{\boldsymbol{X}} = {C_o }({\boldsymbol{A}}),
\end{eqnarray}
where ${C_o }(\cdot)$ denotes the ASC function, $o$ is the compression ratio. 
The compression procedure is illustrated in the following definition.
\begin{definition}\label{sc}
	To unify the ASC expressions for various semantic communication systems, we define the process of ASC as
	\begin{eqnarray}\label{compression_equation}
		{{\boldsymbol{X}}^{_k}} = \left\{ {\begin{array}{*{20}{c}}
				{{{\boldsymbol{A}}{^k}},\;\;{\omega _k^c} \ge {\omega _0}}\\
				{0,\;\;\;\;\;{\omega _k^c} < {\omega _0}}
		\end{array}} \right.
	\end{eqnarray}
	where ${{\boldsymbol{A}}^{_k}}$ is the $k$-th semantic feature and ${\omega _k^c}$ is the importance weight of $k$-th feature for semantic concept $c$, which will be detailed in Section \ref{subsec:evaluation}. ${\omega _0}$ is the importance weight threshold that determined by the compression ratio. 
\end{definition}
Equation (\ref{compression_equation}) indicates that if a semantic feature's importance weight exceeds the threshold, it will be transmitted; otherwise, it will be discarded. Existing semantic communication systems can be regarded as a special case when ${\omega _0} =0$.  

Next, the compressed semantics is encoded by channel encoder to generate symbols for transmission, which can be denoted by
\begin{eqnarray}\label{}
	{\boldsymbol{M}} = {Q_{\boldsymbol{\sigma}}}({\boldsymbol{X}}),
\end{eqnarray}
where $Q_{\boldsymbol{\sigma}}(\cdot)$ denotes the channel encoder network with parameter set ${\boldsymbol{\sigma}}$.

Then, the encoded symbols are transmitted via a wireless channel, and the received signal is expressed as
\begin{eqnarray}\label{}
{\boldsymbol{Y}} = h{\boldsymbol{M}} + \boldsymbol{n},
\end{eqnarray}
where $h$ denotes the channel gain, $\boldsymbol{n}$ is a vector sampled from Gaussian distribution.

We consider the transmitted symbols are mapped into bits by binary quantization, and thus transmission in the physical layer still follows Shannon's classic information theory, and the transmission rate of user $i$ is
\begin{eqnarray}\label{R}
	{R_i} = {B_i}{\rm{lo}}{{\rm{g}}_{\rm{2}}}{\rm{(1 + }}\frac{{{h_i}{P_i}}}{{{N_0}{B_i}}}{\rm{)}},
\end{eqnarray}
where $B_i$ is the bandwidth of user $i$, $P_i$ is the transmission power of user $i$, $h_i$ is the channel gain between user $i$ and edge server, and $N_0$ is the noise power spectral density.

Denoting the initial data size of semantic information that users extracts is $d_0$, and the semantic compression ratio of user $i$ is $o_i$, the amount of data actually transmitted by user $i$ is ${d_i} = {d_0} \times ({1-o_i})$. Therefore, the transmission delay of user $i$ is 
\begin{eqnarray}\label{t}
{t_i} = \frac{{{d_i}}}{{{R_i}}}.
\end{eqnarray}

In actual scenarios (e.g., Internet of Vehicles (IoV)), a large number of tasks are latency-sensitive and thus there is always a strict transmission delay constraint, which can be denoted by $t_0$. Thus, the success transmission probability of user $i$ is ${\rm{P}}({t_i} \le {t_0})$. To calculate ${\rm{P}}({t_i} \le {t_0})$, we have the following lemma.

\noindent\textbf{$Lemma\ 1.$} The success transmission probability of user $i$ is 
\begin{align}\label{transmission}
	P({t_i} \le {t_0}) = 2Q\left( {\frac{{{2^{{a_i}(1 - {o_i})}} - 1}}{{{b_i}\delta }}} \right)
\end{align}
where ${a_i} = \frac{{{d_0}}}{{{B_i}{t_0}}}$, ${b_i} = \frac{{{P_i}}}{{{N_0}{B_i}}}$ and ${\delta ^2}$ is variance of the channel gain. The Q-function is the tail distribution function of the standard normal distribution.

\begin{proof}
	Please see Appendix \ref{proof_lemma1}.
\end{proof}
\begin{remark}
	As we observe from Lemma 1, the success transmission probability is mainly affected by power, bandwidth, and semantic compression ratio. Therefore, the success transmission probability can be improved by optimizing the semantic compression ratio and resource allocation.
\end{remark}

Then, received symbols are decoded to recover semantics via channel decoder, which can be expressed as
\begin{eqnarray}\label{}
	{{{\boldsymbol{X}}'}} = {Q_{\boldsymbol{\chi }}^{ - 1}}({\boldsymbol{Y}}),
\end{eqnarray}
where ${Q_{\boldsymbol{\chi }}^{ - 1}}(\cdot)$ denotes the channel decoder network with parameter set ${\boldsymbol{\chi }}$.

Finally, the semantic receiver inputs the recovered semantics ${{\boldsymbol{X}}'}$ into semantic decoder to complete the intelligent tasks. Specifically, the output is
\begin{eqnarray}\label{}
	{\boldsymbol{p}} = {S_{\boldsymbol{\mu }}^{ - 1}}({\boldsymbol{{{\boldsymbol{X}}'}}}),
\end{eqnarray}
where ${\boldsymbol{p}}$ is the task result, which will be returned to the transmitter and ${S_{\boldsymbol{\mu }}^{ - 1}}(\cdot)$ denotes the semantic decoder with the parameter set ${\boldsymbol{\mu }}$. 


\subsection{Problem Formulation}
In task-oriented semantic communications, conventional communication metrics that ignore the underlying meaning of the source are no longer applicable, and thus new performance metrics need to be investigated at the semantic level. To simultaneously evaluate the impact of transmission and ASC on the performance of semantic communications, we define a novel metric, namely the success probability of tasks, which is expressed in the following definition.
\begin{definition}\label{sc}
	The success probability of tasks of $i$-th user can be expressed as
	\begin{equation}\label{EA}
	    \Phi_i = \eta ({o_i}) \times P({t_i} \le {t_0}).
	\end{equation}
	where $\eta ({o_i})$ is the probability that task is successfully executed under success transmission, while the compression ratio is $o_i$.
\end{definition}
From (\ref{EA}), we see that the proposed success probability of tasks used to evaluate the semantic communication performance can control the tradeoff between the semantic transmission and the semantic understanding. We consider service level agreement (SLA), where users are prioritized with different service levels according to their objectives. For example, in the smart factory scenario, users who perform fire detection are typically characterized by a higher service level due to the requirements of high-reliable and low-latency communication \cite{3gpp}. Note that SLA is a general method existing in most modern cellular systems, such as 5G\cite{3gpp}. Consider a set ${{\cal N}}$ of $N$ service levels according to users' tasks. Let the weight of service level $n$ be ${{\varepsilon _n}}$ and ${r_{in}} \in \left\{ {0,1} \right\}$ denote the user association index, i.e., ${r_{in}} = 1$ means that the user $i$ belongs to service level $n$; otherwise, we have ${r_{in}} = 0$. Therefore, the importance weight of user $i$ is expressed as
\begin{eqnarray}
	{w_i} = \sum\limits_{n = 1}^N {{\varepsilon _n}{r_{in}}}.
\end{eqnarray}
Considering SLA, the weighted sum success probability of tasks of the whole semantic communication system is expressed as
\begin{equation}
	\Phi = \sum\limits_{i = 1}^U {{\beta _i}{w_i}{\Phi _i}}.
\end{equation}
where ${\beta _i} \in \left\{ {0,1} \right\}$ denotes the user selection index, i.e., ${\beta _i} = 1$ indicates user $i$ is selected; while otherwise, we have ${\beta _i} = 0$. It is extremely necessary to ensure the performance of users with higher priority, so only a subset of users may be selected to complete intelligent tasks due to the limited wireless resources.

We aim to optimize the resource allocation, compression ratios, and user selection simultaneously to maximize the weighted sum success probability of tasks of the semantic system under the resource constraints. Mathematically, we formulate the optimization problem as
\begin{align}\label{Q1}
		& \mathop {\max }\limits_{{\boldsymbol{B}},{\boldsymbol{P}},{\boldsymbol{o}},{\boldsymbol{\beta}}} \Phi \\
		\rm{s.t.}\;\; & \sum\limits_{i = 1}^U {{\beta _i}{B_i}}  \le {B_{\max }}, \tag{\theequation a}\\
		&{B_i} \ge {B_{\min }}, \forall i \in {{\cal U}_s}, \tag{\theequation b}\\
		& \sum\limits_{i = 1}^U {{\beta _i}{P_i}}  \le {P_{\max }}, \tag{\theequation c}\\
		&{P_i} \ge {P_{\min }}, \forall i \in {{\cal U}_s}, \tag{\theequation d}\\
		&0 < {o_i} < 1, \forall i \in {{\cal U}},\tag{\theequation e}\\
		& {\beta _i} \in \left\{ {0,1} \right\},\forall i \in {{\cal U}}, \tag{\theequation f}\\
		& \sum\limits_n {{r_{in}} = 1},\forall i \in {{\cal U}}, \tag{\theequation g}
\end{align}
where ${\boldsymbol{\beta }} = [{\beta _1},{\beta _2}, \cdots ,{\beta _i}, \cdots {\beta _U}]$, ${{\cal U}_s}$ is the set of selected users, $B_{\min }$ is the minimum bandwidth allocated to users, $B_{\max }$ is the maximum total bandwidth, $P_{\min }$ is the minimum transmit power allocated to users, and $P_{\max }$ is the maximum total transmit power. Constraint (\ref{Q1}a) indicates that the sum bandwidth of selected users cannot exceed a given threshold, which refers to the system bandwidth. Constraints (\ref{Q1}b) and (\ref{Q1}d) are the minimum bandwidth and the minimum transmit power constraints, respectively. Constraint (\ref{Q1}c) indicates that the sum transmit power of selected users cannot exceed a given value, which guarantees that the energy consumption of the whole system is limited. Constraint (\ref{Q1}e) is the compression ratio constraint. Constraint (\ref{Q1}g) indicates that each user can only belong to one service level. The main notations of this paper are summarized in Table \ref{tab:notation}.

\begin{table}[t]
	\normalsize
	\centering
	\caption{List of Notations}
	\rowcolors{1}{blue!15}{}
	\label{tab:notation}
	\begin{tabular}{|c|m{6.4cm}<{\centering}|}
		\hline
		\textbf{Notation} & \textbf{Description} \\
		\hline
		$U$ & Number of users\\
		\hline
		$N$ & Number of service levels\\
		\hline
		${{\varepsilon _n}}$ & Weight of service level $n$\\
		\hline
		${r_{in}}$ & User association index\\
		\hline
		$w_i$ & Importance weight of user $i$\\
		\hline
		$\omega _k^c$ & Importance weight of $k$ semantic feature\\
		\hline
		$\boldsymbol{p}$ & Task result\\
		\hline
		$\boldsymbol{y_l}$ & Task label\\
		\hline
		$\kappa$ & Weight for the mutual information\\
		\hline
		$R_i$ & Transmission rate of user $i$\\
		\hline
		$B_i$ & Bandwidth of user $i$\\
		\hline
		$P_i$ & Transmission power of user $i$\\
		\hline
		$N_0$ & Noise power spectral density\\
		\hline
 		${{d_0}}$ & Initial data size of extracted semantics \\
 		\hline
 		$o_i$ & Semantic compression ratio of user $i$\\
 		\hline
 		$t_i$ & Transmission delay of user $i$\\
 		\hline
 		$t_0$ & Transmission delay constraint\\
 		\hline
 		${\Phi _i}$ & success probability of tasks of user $i$\\
 		\hline
 		$\Phi$ & success probability of tasks of the whole semantic communication system\\
 		\hline
 		$B_{\min }$ & Minimum bandwidth allocated to users\\
 		\hline
 		$B_{\max }$ & Maximum total bandwidth\\
 		\hline
 		$P_{\min }$ & Minimum transmit power allocated to users\\
 		\hline
 		$P_{\max }$ & Maximum total transmit power\\
 		\hline
 		$\boldsymbol{\beta}$ & User selection vector\\
		\toprule
	\end{tabular}
	\vspace{-0.cm}
\end{table}

\vspace{-0.cm}
\section{E2E Semantic Communication and ASC}
\label{sec:details}
In this section, we first detail the architecture and loss function of the proposed E2E semantic communication system. Then, we illustrate the method to evaluate the importance of semantic features, which is the basis of the ASC approach. Finally, we investigate the relationship between intelligent task performance and compression ratio by a numerical method.
\subsection{E2E Semantic Communication Network Design}
\label{subsec:Implementation}
In this work, we target various intelligent tasks including source data with complex features, such as image recognition. To capture and transmit the meaningful semantics of those source data, deep neural networks (DNNs) are used to implement the joint encoder-decoder framework. The superiority of DNNs in encoding and decoding has been verified that they can outstandingly support the transmission of different types of source data\cite{Xie_Deep,retrival1}.
There are three types of DNNs used for semantic encoder and decoder mostly, including recurrent neural networks (RNN), CNN, and fully-connected neural networks (FCN). In general, CNN is more suitable for image data, while RNN is more suitable for time-series data (e.g., audio and text). 
The DNN-based E2E semantic communication framework comprises four parts: semantic encoder, channel encoder, channel decoder, and semantic decoder. The first and second are employed in the transmitter that identifies the task-relevant features from the raw data and maps the feature values to the channel input symbols, respectively. The third and the last are employed in the receiver, aiming at symbol detection and semantic reconstruction, respectively. 

The semantic transceiver is jointly trained in an end-to-end manner, where the gradients require to be backpropagated from the output layer of the semantic decoder to the input layer of the semantic encoder. To do so, the communication channel also needs to be modeled by a neural network and employed between models of channel encoder and decoder. In this work, we focus on the ASC and the following semantic resource allocation, thus considering the basic additive white Gaussian noise (AWGN) channel similar to \cite{Xie_Deep}. To achieve ASC, the DNN-based semantic encoder and decoder are trained according to a loss function that can intuitively represent the semantic importance of the feature. The loss function will be presented in the next part.

\subsection{Loss Function Design}
\label{subsec:loss}
The important goal of designing a task-oriented semantic communication system is to maximize the intelligent task performance and the capacity or the data transmission rate simultaneously. Compared with the bit error rate, the mutual information can provide extra information to train a transceiver\cite{Xie_Deep}. The mutual information of the transmitted semantics, $\boldsymbol{X}$, and the received semantics, $\boldsymbol{Y}$, can be computed by
\begin{eqnarray}\label{}
	\begin{aligned}
		I\left( {{\boldsymbol{X}};{\boldsymbol{Y}}} \right) = &{{\rm E}_{p\left( {{\boldsymbol{x}},{\boldsymbol{y}}} \right)}}\left[ {\log \frac{{p\left( {{\boldsymbol{x}},{\boldsymbol{y}}} \right)}}{{p\left( {\boldsymbol{x}} \right)p\left( {\boldsymbol{y}} \right)}}} \right]\\
		= &{{\rm E}_{p\left( {{\boldsymbol{x}},{\boldsymbol{y}}} \right)}}\left[ {\log p\left( {{\boldsymbol{y}}\left| {\boldsymbol{x}} \right.} \right) - \log p\left( {\boldsymbol{y}} \right)} \right],
	\end{aligned}
\end{eqnarray}
where $\left( {{\boldsymbol{X}},{\boldsymbol{Y}}} \right)$ is a pair of random variables with values over the space ${{\cal X}} \times {{\cal Y}}$, where ${\cal X}$ and ${\cal Y}$ are the spaces for ${\boldsymbol{X}}$ and ${\boldsymbol{Y}}$. $p\left( {\boldsymbol{x}} \right)$ and $p\left( {\boldsymbol{y}} \right)$ are the marginal probability of sent ${\boldsymbol{X}}$ and received ${\boldsymbol{Y}}$, respectively, and $p\left( {{\boldsymbol{x}},{\boldsymbol{y}}} \right)$ is the joint probability of ${\boldsymbol{X}}$ and ${\boldsymbol{Y}}$. To effectively estimate the mutual information of ${\boldsymbol{X}}$ and ${\boldsymbol{Y}}$, the following theorem is presented.
\begin{theorem}\label{mi}
	The upper bound of the mutual information can be expressed as
	\begin{eqnarray}\label{up}
		\begin{aligned}
		{I_{{\rm{up}}}}\left( {{\boldsymbol{X}};{\boldsymbol{Y}}} \right):=&{{\rm{E}}_{p\left( {{\boldsymbol{x}},{\boldsymbol{y}}} \right)}}\left[ {\log p\left( {{\boldsymbol{y}}\left| {\boldsymbol{x}} \right.} \right)} \right]\\
		-&{{\rm{E}}_{p\left( {\boldsymbol{x}} \right)}}{{\rm{E}}_{p\left( {\boldsymbol{y}} \right)}}\left[ {\log p\left( {{\boldsymbol{y}}\left| {\boldsymbol{x}} \right.} \right)} \right]
		\end{aligned}
	\end{eqnarray}
\end{theorem}
\begin{proof}
	In order to prove that ${I_{{\rm{up}}}}\left( {{\boldsymbol{X}};{\boldsymbol{Y}}} \right)$ is the upper bound of mutual information, it is only necessary to prove that ${I_{{\rm{up}}}}\left( {{\boldsymbol{X}};{\boldsymbol{Y}}} \right)$ is greater than the true mutual information $I\left( {{\boldsymbol{X}};{\boldsymbol{Y}}} \right)$. The gap between the true mutual information and the upper bound can be denoted by
	\begin{align}\label{}
			\Delta : = &{I_{{\rm{up}}}}\left( {{\boldsymbol{X}};{\boldsymbol{Y}}} \right) - I\left( {{\boldsymbol{X}};{\boldsymbol{Y}}} \right)\nonumber\\
			= &{{\rm{E}}_{p\left( {{\boldsymbol{x}},{\boldsymbol{y}}} \right)}}\left[ {\log p\left( {{\boldsymbol{y}}\left| {\boldsymbol{x}} \right.} \right)} \right] - {{\rm{E}}_{p\left( {\boldsymbol{x}} \right)}}{{\rm{E}}_{p\left( {\boldsymbol{y}} \right)}}\left[ {\log p\left( {{\boldsymbol{y}}\left| {\boldsymbol{x}} \right.} \right)} \right]\nonumber\\
			- &{{\rm E}_{p\left( {{\boldsymbol{x}},{\boldsymbol{y}}} \right)}}\left[ {\log p\left( {{\boldsymbol{y}}\left| {\boldsymbol{x}} \right.} \right) - \log p\left( {\boldsymbol{y}} \right)} \right]\nonumber\\
			= &{{\rm{E}}_{p\left( {{\boldsymbol{x}},{\boldsymbol{y}}} \right)}}\left[ {\log p\left( {\boldsymbol{y}} \right)} \right] - {{\rm{E}}_{p\left( {\boldsymbol{x}} \right)}}{{\rm{E}}_{p\left( {\boldsymbol{y}} \right)}}\left[ {\log p\left( {{\boldsymbol{y}}\left| {\boldsymbol{x}} \right.} \right)} \right]\nonumber\\
			= &{{\rm{E}}_{p\left( {\boldsymbol{y}} \right)}}\left[ {\log p\left( {\boldsymbol{y}} \right) - {{\rm{E}}_{p\left( {\boldsymbol{x}} \right)}}\left[ {\log p\left( {{\boldsymbol{y}}\left| {\boldsymbol{x}} \right.} \right)} \right]} \right]\nonumber\\
			= &{{\rm{E}}_{p\left( {\boldsymbol{y}} \right)}}\left[ {\log \left[ {{{\rm{E}}_{p({\boldsymbol{x}})}}\left[ {p\left( {{\boldsymbol{y}}\left| {\boldsymbol{x}} \right.} \right)} \right]} \right] - {{\rm{E}}_{p\left( {\boldsymbol{x}} \right)}}\left[ {\log p\left( {{\boldsymbol{y}}\left| {\boldsymbol{x}} \right.} \right)} \right]} \right] \ge 0
	\end{align}
	The last step is derived from Jensen's inequality. This completes the proof of Theorem \ref{mi}.
\end{proof}
However, the conditional relation $p\left( {{\boldsymbol{y}}\left| {\boldsymbol{x}} \right.} \right)$ between variables in Theorem \ref{mi} is unavailable, and a variational distribution ${q_\theta }({\boldsymbol{y}}\left| {\boldsymbol{x}} \right.)$ with parameter $\theta$ is used to  approximate $p\left( {{\boldsymbol{y}}\left| {\boldsymbol{x}} \right.} \right)$. According to the Theorem 3.2 in \cite{club}, minimizing the upper bound on mutual information is equivalent to minimizing $ - {{\rm{E}}_{p\left( {{\boldsymbol{x}},{\boldsymbol{y}}} \right)}}\left[ {\log {q_\theta }\left( {{\boldsymbol{y}}\left| {\boldsymbol{x}} \right.} \right)} \right]$. With samples $\left\{ {\left( {{{\boldsymbol{x}}_i},{{\boldsymbol{y}}_i}} \right)} \right\}_{i = 1}^L$, we can minimize the log-likelihood function ${L_{{\rm{MI}}}}(\theta ): =  - \frac{1}{L}\sum\limits_{i = 1}^L {\log {q_\theta }({{\boldsymbol{y}}_i}\left| {{{\boldsymbol{x}}_i}} \right.)}$,  which is the unbiased estimation of $ - {{\rm{E}}_{p\left( {{\boldsymbol{x}},{\boldsymbol{y}}} \right)}}\left[ {\log {q_\theta }\left( {{\boldsymbol{y}}\left| {\boldsymbol{x}} \right.} \right)} \right]$. In this paper, the variational distribution ${q_\theta }({\boldsymbol{y}}\left| {\boldsymbol{x}} \right.)$ is implemented with neural networks and minimized via gradient-descent method.

The parameters of the semantic communication network are optimized via the following loss function
\begin{eqnarray}\label{loss}
L({\boldsymbol{y_l}},{\boldsymbol{p}};{\boldsymbol{\alpha}},{\boldsymbol{\mu}}) =  L_{\rm{T}}({\boldsymbol{y_l}},{\boldsymbol{p}})  - \kappa {I_{{\rm{up}}}}\left( {{\boldsymbol{X}};{\boldsymbol{Y}}} \right)
\end{eqnarray}
where ${\boldsymbol{y_l}}$ is the task label. The first term $L_{\rm{T}}({\boldsymbol{y_l}},{\boldsymbol{p}})$ is the loss function related to the task (i.e., cross-entropy for classification task, triplet loss for object detection task, etc.), which aims to maximize the task performance by training the whole system. The second one ${I_{{\rm{up}}}}\left( {{\boldsymbol{X}};{\boldsymbol{Y}}} \right)$ is the estimation of mutual information between ${\boldsymbol{X}}$ and ${\boldsymbol{Y}}$, which maximizes the achieved data rate during the transmitter training. Parameter $\kappa$, between 0 and 1, is the weight for the mutual information.

The training process of the proposed task-oriented semantic communication network consists of two phases due to different loss functions. After initializing the parameters, the first phase is to train the mutual information model by unsupervised learning to estimate the achieved data rate for the second phase. The second phase is to train the whole system with (\ref{loss}) as the loss function. Each phase aims to minimize the loss by gradient descent with mini-batch until the stop criterion is met, the max number of iterations is reached, or none of the terms in the loss function is decreased anymore. Notably, ASC is only performed during inference.

\subsection{Semantic Importance Evaluation}
\label{subsec:evaluation}
In task-oriented semantic communications, different semantic features are of different importance for completing intelligent tasks, and thus there are still semantic redundancies that are irrelevant to the intelligent tasks, which can be further compressed\cite{Lee_SNIP}. Here, the importance of semantics is defined as the correlation between semantics and the task. The way to measure the importance of semantic features can be variable with different semantic communication systems, and here we employ a gradient-based approach. Based on the semantic communication system trained in \ref{subsec:loss}, we first compute the gradient of the activation value for semantic concept $c$ (such as objects, properties, and actions) \cite{semantic_concept}, $y^c$ (before Softmax layer), with respect to $k$-th semantic feature activations ${{\boldsymbol{A}}^{_k}}$, i.e., ${\frac{{\partial {y^c}}}{{\partial {\boldsymbol{A}}^k}}}$. These gradients flowing back are global-average-pooled over the width and height dimensions (indexed by $i$ and $j$ respectively) to obtain the semantic importance weights
\begin{eqnarray}\label{weight}
\omega _k^c = \frac{1}{{W \times H}}\sum\limits_i {\sum\limits_j {\frac{{\partial {y^c}}}{{\partial {\boldsymbol{A}}_{ij}^k}}} }
\end{eqnarray}
where $W$ and $H$ are the width and height of ${{\boldsymbol{A}}^{_k}}$, and ${\boldsymbol{A}}_{ij}^k$ is the activation value at the $i$-th row and the $j$-th column of the feature map. During computation of $\omega _k^c$ while backpropagating gradients with respect to activations, the exact computation amounts to successive matrix products of the weight matrices and the gradient with respect to activation functions till the final convolution layer that the gradients are being propagated to. Hence, this weight $\omega _k^c$ represents a partial linearization of the deep network downstream from ${{\boldsymbol{A}}}$, and captures the ‘\emph{semantic importance}’ of semantic feature $k$ for a semantic concept $c$\cite{CAM}.

Since the importance weights are only related to network parameters, these weights can be regarded as shared knowledge and stored in the knowledge base of the sender and receiver, where the knowledge base could be different for various tasks. Consequently, there is no need to transmit the indices corresponding to the transmitted feature maps in the subsequent semantic communication process. In this work, we only calculate the semantic importance weights via a gradient-based method. One can easily extend the proposed ASC method to other calculation methods such as attention-based mechanisms\cite{Wang_Attention}. Based on the obtained semantic importance weights, ASC can be performed consequently. ASC proposed in this paper has two major benefits: first, it lessens the requirements of subsequent computing resources; second, it dramatically reduces the amount of the transmitted data, hence reducing the demand for communication resources and transmission delay.

\subsection{Intelligent Task Performance Model}
\label{subsec:intelligent}
To solve the problem (\ref{Q1}), we have to investigate the intelligent task performance model first, i.e., $\eta (o)$, which draws the success probability of task under compression ratio $o$ and success transmission and is closely related to the objective function $\Phi$. However, deriving a close-form expression for $\eta (o)$ is intractable due to the inexplicability of neural networks. In this subsection, we find the relationship between the semantic compression ratio and task performance by approximating a function to the statistics of the model evaluation. Note that, in practice, training the E2E model is executed in the server, thus calculating such function using a numerical approach is possible.

To obtain a point set ${{\cal D}}$ of $D$ points reflecting the mapping between task performance $\eta$ and compression ratio $o$, we first calculate the importance weights of the feature maps via (\ref{weight}), and then remove the unimportant feature maps in turn and calculate the corresponding task performance and compression ratio. Inspired by the ideas in \cite{fitting}, we empirically find that the points of ${{\cal D}}$ can be estimated by an exponential function, i.e., ${\eta}(o) = {\zeta _1}{e^{{\zeta _2}o}} + {\zeta _3}{e^{{\zeta _4}o}}$. Then, we learn the parameters ${\boldsymbol{\zeta }} = [{\zeta _1},{\zeta _2},{\zeta _3},{\zeta _4}]$ via a numerical approach, which vary with the adopted neural networks. The parameters solving algorithm based on gradient descent is summarized in \textbf{Algorithm} \ref{fitting}.


\begin{algorithm}[t]
	\normalsize
	\caption{Parameters Solving Algorithm.}
	\begin{algorithmic}[1]
		\STATE \textbf{Input:} Initialize parameters ${{\boldsymbol{\zeta }}} = [\zeta _1,\zeta _2,\zeta _3,\zeta _4]$, point set ${{\cal D}}$, step length $\delta$, threshold $L_0$.
		\REPEAT		
		\FOR {$(o^d,\eta^{d*}) \in {{\cal D}}$}
		\STATE Compute ${\eta^d}(o^d) = {\zeta _1}{e^{{\zeta _2}o^d}} + {\zeta _3}{e^{{\zeta _4}o^d}}$.
		\ENDFOR
		\STATE Compute the loss $L({\boldsymbol{\zeta }}) = \frac{1}{2D}{\sum\limits_{d = 1}^D {\left( {\eta _{}^d({o^d}) - \eta^{d*}} \right)} ^2}$.
		\STATE Compute the gradient of ${\boldsymbol{\zeta }}$: $G({\boldsymbol{\zeta }}) = \frac{{\partial L({\boldsymbol{\zeta }})}}{{\partial {\boldsymbol{\zeta }}}}$.
		\STATE Update parameters ${\boldsymbol{\zeta }}: = {\boldsymbol{\zeta }} - \delta G({\boldsymbol{\zeta }})$.
		\UNTIL {$L({\boldsymbol{\zeta }}) \le {L_0}$}
		\STATE \textbf{Output:} ${{\boldsymbol{\zeta }}} = [\zeta _1,\zeta _2,\zeta _3,\zeta _4]$.
	\end{algorithmic}
	\label{fitting}
\end{algorithm}

\section{CRRA Algorithm}
\label{sec:CRRA}
Based on the above analysis and results, we then focus on solving the problem (\ref{Q1}) in the following two sections. In this section, we consider the scenarios where users perform the tasks with similar priority (e.g., pedestrian and vehicle detection in IoV), and thus users have the same service levels, and all of them should be selected, i.e., $w_i=1, \beta _i = 1, \forall i \in \cal U$. In the considered scenarios, the optimization problem (\ref{Q1}) is first simplified to a resource allocation and compression ratios optimization problem to maximize the total success probability of tasks. Then, the CRRA algorithm is proposed to solve the optimization problem. 


Based on the approximation of $Q$-function $Q(x) \approx \frac{1}{2}{e^{ - \frac{{{x^2}}}{2}}}$ \cite{Q_bound}, problem (\ref{Q1}) can be reformulated as 
\begin{align}\label{Q3}
	&\mathop {\max }\limits_{{\boldsymbol{B}},{\boldsymbol{P}},{\boldsymbol{o}}} \sum\limits_{i = 1}^U {g_i \times } {{\eta}\left( {{o_i}} \right)}\\
	\rm{s.t.}\;\;&(\ref{Q1}a)-(\ref{Q1}f),\nonumber
\end{align}
where $g_i = {\exp\left\{{ - \frac{1}{2}{{\left[ {\frac{{{N_0}{B_i}\left[ {{2^{\left[ {\frac{{{d_0}\left( {1 - {o_i}} \right)}}{{{B_i}{t_0}}}} \right]}} - 1} \right]}}{{\delta {P_i}}}} \right]}^2}} \right\}}$.

Since the objective function is not concave, the total success probability of tasks maximization problem (\ref{Q3}) is non-convex, hence, it is generally hard to optimize resource allocation and compression ratios directly. To solve the problem (\ref{Q3}), we divide it into two subproblems and then solve these two subproblems iteratively. In particular, we first fix the resource allocation and calculate the optimal compression ratio for each user. Then, resource allocation problem is formulated and solved with the obtained compression ratios. The two subproblems are iteratively solved until a convergent solution is obtained. 

\vspace{-0.2cm}
\subsection{Optimal Compression Ratios}
Given the resource allocation policy, (\ref{Q3}) can be simplified as 
\begin{align}\label{Q4}
	&\mathop {\max }\limits_{{\boldsymbol{o}}} \sum\limits_{i = 1}^U {g_i \times } {{\eta}\left( {{o_i}} \right)}\\
	\rm{s.t.}\;\;\;&0 < {o_i} < 1,\forall i \in {{\cal U}}.\tag{\theequation a}
\end{align}

We can observe from (\ref{Q4}) that if the resource allocation policy is fixed, the optimal semantic compression ratio of each user is independent. Thus, our goal transforms into maximizing each user's success probability of tasks. For user $i$, the problem is 
\begin{align}\label{Q5}
&\mathop {\max }\limits_{{o_i}} g_i \times {{\eta}\left( {{o_i}} \right)}\\
\rm{s.t.}\;\;\;&0 < {o_i} < 1,\forall i \in {{\cal U}}.\tag{\theequation a}
\end{align}

Considering the range of $o_i$ is within 0 and 1, we here employ the one-dimension enumeration method to obtain the optimal semantic compression ratio. The algorithm for solving problem (\ref{Q4}) is summarized in \textbf{Algorithm} \ref{compression}.

\begin{algorithm}[t]
	\normalsize
	\caption{Compression ratio optimization with one-dimension enumeration method.}
	\begin{algorithmic}[1]
		\STATE \textbf{Input:} ${\boldsymbol{B}}$, ${\boldsymbol{P}}$, $h_0 = 0$.
		\FOR {i = 1:U}
		\FOR {$o_i$ = 0.01:0.01:1}
		\STATE Compute $h=g_i \times {{\eta}\left( {{o_i}} \right)}$.
		\IF {$h \ge {h_0}$}
		\STATE $h_0 = h$ and $\boldsymbol{o}_{opt}(i) = o_i$.
		\ENDIF 
		\ENDFOR
		\ENDFOR
		\STATE \textbf{Output:} $\boldsymbol{o}_{opt}$.
	\end{algorithmic}
	\label{compression}
\end{algorithm}

\vspace{-0.2cm}
\subsection{Optimal Resource Allocation}
With the obtained semantic compression ratios, we then optimize the bandwidth and power of the considered semantic communication systems. Note that given $o_i$, ${{\eta}\left( {{o_i}} \right)}$ can be seen as a constant, which is denoted by ${\alpha _i}$. Thus, the resource allocation problem can be reformulated as
\begin{align}\label{Q6}
	&\mathop {\min }\limits_{{\boldsymbol{B}},{\boldsymbol{P}}} \sum\limits_{i = 1}^U { - {\alpha _i} \times g_i} \\
	\rm{s.t.}\;\;&(\ref{Q1}a)-(\ref{Q1}d).\nonumber
\end{align}

To solve problem (\ref{Q6}), we first convert the non-convex problem into a convex optimization problem. In particular, by introducing slack variables ${\boldsymbol{f}} = [{f_1},{f_2},...,{f_U}]$, ${\boldsymbol{l}} = [{l_1},{l_2},...,{l_U}]$, ${\boldsymbol{x}} = [{x_1},{x_2},...,{x_U}]$, ${\boldsymbol{m}} = [{m_1},{m_2},...,{m_U}]$ and ${\boldsymbol{q}} = [{q_1},{q_2},...,{q_U}]$, problem (\ref{Q6}) can be transformed into
\begin{align}\label{Q7}
	&\mathop {\min }\limits_{{\boldsymbol{B}},{\boldsymbol{P}},{\boldsymbol{f}},{\boldsymbol{l}},{\boldsymbol{x}},{\boldsymbol{m}},{\boldsymbol{q}}} \sum\limits_{i = 1}^U { - {\alpha _i} \times {f_i}} \\
	\rm{s.t.}\;\;&{f_i} \le {e^{{l_i}}},\forall i \in {{\cal U}}, \tag{\theequation a}\\
	{\rm{        }}&{l_i} \le  - \frac{1}{2}x_i^2,\forall i \in {{\cal U}}, \tag{\theequation b}\\
	{\rm{        }}&{x_i} \ge \frac{{{N_0}{B_i}{m_i}}}{{\delta {P_i}}},\forall i \in {{\cal U}}, \tag{\theequation c}\\
	{\rm{       }}&{m_i} \ge {2^{{q_i}}} - 1,\forall i \in {{\cal U}}, \tag{\theequation d}\\
	{\rm{       }}&{q_i} \ge \frac{{{d_0}\left( {1 - \sigma } \right)}}{{{B_i}{t_i}}},\forall i \in {{\cal U}}, \tag{\theequation e}\\
	{\rm{       }}&(\ref{Q1}a)-(\ref{Q1}d).\nonumber
\end{align}
However, constraints (\ref{Q7}a) and (\ref{Q7}c) are still non-convex. 

For constraint (\ref{Q7}a), we use the successive convex approximation (SCA) method to transform it into a convex constraint. Performing a first-order Taylor expansion of ${e^{{l_i}}}$ at ${e^{l_i^j}}$, then we have
\begin{eqnarray}\label{}
{f_i} \le {e^{l_i^j}} + \left( {{l_i} - l_i^j} \right){e^{l_i^j}}, 
\end{eqnarray}
while the superscript $j$ represents the value obtained after $j$-th iteration of the variable. 

For constraint (\ref{Q7}c), slack variable ${\boldsymbol{z}} = [{z_1},{z_2},...,{z_U}]$ is introduced, and have 
\begin{eqnarray}\label{z}
{z_i} \ge {B_i}{m_i}.
\end{eqnarray} 
Thus, constraint (\ref{Q7}c) can be transformed into 
\begin{eqnarray}\label{xp}
{x_i}{P_i} \ge \frac{{{N_0}{z_i}}}{{{\delta _i}}}.
\end{eqnarray}
(\ref{z}) can be rewritten as 
\begin{eqnarray}\label{}
{z_i} \ge {B_i}{m_i} = \frac{1}{4}\left( {{{\left( {{B_i} + {m_i}} \right)}^2} - {{\left( {{B_i} - {m_i}} \right)}^2}} \right)\end{eqnarray}
By performing a first-order Taylor expansion of ${\left( {{B_i} - {m_i}} \right)^2}$ at point $\left( {B_i^j,m_i^j} \right)$
and using SCA, we have 
\begin{eqnarray}\label{z_i1}
	\begin{aligned}
		{z_i} \ge \frac{1}{4}( {{\left( {{B_i} + {m_i}} \right)}^2} - 2\left( {{B_i} - {m_i}} \right)\left( {{B_i}^j - {m_i}^j} \right) \\ + {{\left( {{B_i}^j - {m_i}^j} \right)}^{\rm{2}}} ). 
	\end{aligned}
\end{eqnarray}
Similarly, (\ref{xp}) is equivalent to
\begin{eqnarray}\label{}
 {x_i}{P_i}{\rm{ = }}\frac{{\rm{1}}}{{\rm{4}}}\left( {{{\left( {{x_i} + {P_i}} \right)}^2} - {{\left( {{x_i} - {P_i}} \right)}^2}} \right) \ge \frac{{{N_0}{z_i}}}{{{\delta _i}}}. 
\end{eqnarray}
By performing a first-order Taylor expansion of ${\left( {{x_i} + {P_i}} \right)^2}$ and ${\left( {{x_i} - {P_i}} \right)^2}$ at point $\left( {{x_i}^j,{P_i}^j} \right)$ and using SCA, we can obtain 
\begin{eqnarray}\label{z_i2}
	\begin{aligned}
		\frac{{4{N_0}{z_i}}}{{{\delta _i}}} \le 2\left( {{x_i} + {P_i}} \right)*\left( {{x_i}^j + {P_i}^j} \right) - {\left( {{x_i}^j + {P_i}^j} \right)^2} \\- 2\left( {{x_i} - {P_i}} \right)*\left( {{x_i}^j - {P_i}^j} \right) + {\left( {{x_i}^j + {P_i}^j} \right)^2}.
	\end{aligned}
\end{eqnarray}
So far, all constraints are transformed into convex, and the optimization problem can be reformulated as
\begin{align}\label{Q8}
	&\mathop {\min }\limits_{{\boldsymbol{B}},{\boldsymbol{P}},{\boldsymbol{f}},{\boldsymbol{l}},{\boldsymbol{x}},{\boldsymbol{m}},{\boldsymbol{q}},{\boldsymbol{z}}} \sum\limits_{i = 1}^U { - {\alpha _i} \times {f_i}} \\
	\rm{s.t.}\;\;&{f_i} \le {e^{l_i^j}} + \left( {{l_i} - l_i^j} \right){e^{l_i^j}},\forall i \in {{\cal U}}, \tag{\theequation a}\\
	{\rm{        }}&{l_i} \le  - \frac{1}{2}x_i^2,\forall i \in {{\cal U}}, \tag{\theequation b}\\
	{\rm{       }}&{m_i} \ge {2^{{q_i}}} - 1,\forall i \in {{\cal U}}, \tag{\theequation c}\\
	{\rm{       }}&{q_i} \ge \frac{{{d_0}\left( {1 - \sigma } \right)}}{{{B_i}{t_i}}},\forall i \in {{\cal U}}, \tag{\theequation d}\\
	{\rm{       }}&(\ref{Q1}a)-(\ref{Q1}d), (\ref{z_i1}), (\ref{z_i2}).\nonumber
\end{align}

Problem (\ref{Q8}) is a convex optimization problem, and can be solved via the dual method\cite{convex}. Optimal results can be obtained by setting the initial value of $l_i^j$, $B_i^j$, $m_i^j$, $x_i^j$ and $P_i^j$, updating variables, and performing iterations until the problem converges, which is summarized in \textbf{Algorithm} \ref{SCA}.
\begin{algorithm}[t]
	\normalsize
	\caption{Resource Allocation with SCA.}
	\begin{algorithmic}[1]
		\STATE Initialize ${\boldsymbol{B}^{(0)}},{\boldsymbol{P}^{(0)}},{\boldsymbol{f}^{(0)}},{\boldsymbol{l}^{(0)}},{\boldsymbol{x}^{(0)}},{\boldsymbol{m}^{(0)}},{\boldsymbol{q}^{(0)}},{\boldsymbol{z}^{(0)}}$. Set iteration number $n=1$.
		\REPEAT
		\STATE Solve convex problem (\ref{Q8}).
		\STATE Denote the optimal solution of (\ref{Q8}) by (${\boldsymbol{B}^{(n)}},{\boldsymbol{P}^{(n)}},{\boldsymbol{f}^{(n)}},{\boldsymbol{l}^{(n)}},{\boldsymbol{x}^{(n)}},{\boldsymbol{m}^{(n)}},{\boldsymbol{q}^{(n)}},{\boldsymbol{z}^{(n)}}$).
		\STATE Set $n = n + 1$
		\UNTIL {the objective value (\ref{Q6}) converges.}
	\end{algorithmic}
	\label{SCA}
\end{algorithm}

Finally, we can iteratively solve (\ref{Q4}) and (\ref{Q8}) until a convergent solution is obtained. The overall CRRA algorithm is summarized in \textbf{Algorithm} \ref{summary}.
\begin{algorithm}[t]
	\normalsize
	\caption{CRRA Algorithm.}
	\begin{algorithmic}[1]
		\STATE Initialize semantic compression ratio $\boldsymbol o$, resource allocation $\boldsymbol B$ and $\boldsymbol P$.
		\REPEAT
		\STATE With fixed resource allocation $\boldsymbol B$ and $\boldsymbol P$, optimize semantic compression ratios $\boldsymbol o$ with the enumeration method.
		\STATE With fixed semantic compression ratios, obtain the optimal resource allocation $\boldsymbol B$ and $\boldsymbol P$ by solving (\ref{Q8}).
		\UNTIL {the objective value (\ref{Q3}) converges.}
	\end{algorithmic}
	\label{summary}
\end{algorithm}

\vspace{-0.1cm}
\section{CRRAUS Algorithm}
\label{sec:CRRAUS}
In this section, we consider the scenarios where users perform the tasks with different priorities (e.g., face and fire detection in the smart factory), and thus users have various service levels, and only part of them can be selected due to the wireless resource constraints. To deal with the problem (\ref{Q1}), the CRRAUS algorithm is proposed, which is able to adaptively adjust the user selection based on the wireless resources and service levels. Then, the convergence and complexity of CRRA and CRRAUS are analyzed. 

\vspace{-0.2cm}
\subsection{Algorithm Design}
It is hard to obtain the optimal solutions to the problem (\ref{Q1}) due to non-concave objective function and nonconvex constraints. To obtain a suboptimal solution to problem (\ref{Q1}), we propose a CRRAUS algorithm, in which problem (\ref{Q1}) is separated into three subproblems and solved iteratively. In particular, we first fix the resource allocation and user selection scheme to calculate the optimal compression ratio for each user. Second, we fix the resource allocation and compression ratios to solve the optimal user selection scheme. Finally, the problem of resource allocation is formulated and solved with the obtained compression ratios and user selection scheme.

Similar to ({\ref{Q4}}) and ({\ref{Q5}}), the first subproblem can be simplified into maximizing each user’s weighted success probability of tasks by optimizing the semantic compression ratio
\begin{align}\label{Q22}
	& \mathop {\max }\limits_{o_i} {{\lambda _i}\Phi _i} \\
	\rm{s.t.}\;\; & 0 \le {o_i} \le 1,\tag{\theequation a}
\end{align}
where ${\lambda _i}{\rm{ = }}{\beta _i}{w _i}$ is a constant only relevant to user $i$.
Problem (\ref{Q22}) can also be solved by the one-dimension enumeration method, and the solution process is omitted.

Given the resource allocation policy and compression ratios, the user selection subproblem can be simplified as
\begin{align}\label{Q23}
	& \mathop {\max }\limits_{\boldsymbol{\beta }} \sum\limits_{i = 1}^U {{\varsigma _i}{\beta _i}} \\
	\rm{s.t.}\;\; & (\ref{Q1}a), (\ref{Q1}c), (\ref{Q1}f), 
\end{align}
where ${\varsigma _i} = {{w _i}\Phi _i}$ can be can be regarded as a constant only related to user $i$. 
Problem (\ref{Q23}) is a 0-1 integer programming problem, which can be solved by branch and bound method\cite{BB}. The algorithm for user selection is summarized in \textbf{Algorithm} \ref{BB}.
\begin{algorithm}[htbp]
	\normalsize
	\caption{User selection with branch and bound method.}
	\begin{algorithmic}[1]
		\STATE Find the optimal solution to the linear programming problem with the 0-1 integer restrictions relaxed.
		\STATE At node 1, let the relaxed solution be the upper bound $Q^{U}$ and the rounded-down integer solution be the lower bound $Q^{L}$. Select the variable with the greatest fractional part for branching.
		\STATE Create two new nodes, one is for the $\beta_j = 0$ and the other is for the $\beta_j = 1$.
		\STATE Solve the relaxed linear programming problem with the new constraint added at each of these nodes, and obtain the relaxed solution $Q$ and $\boldsymbol{\beta}$.
		\STATE Let the upper bound $Q^{U} = Q$ at each node, and the existing maximum integer solution $Q^{L}$ is the lower bound.
		\IF {All elements in $\boldsymbol{\beta}$ are integers}
		\STATE The optimal integer solution $\boldsymbol{\beta^*} = \boldsymbol{\beta}$.
		\ELSE 
		\STATE Branch from the node with the greatest upper bound and return to step 3.
		\ENDIF
	\end{algorithmic}
	\label{BB}
\end{algorithm}

With the obtained compression ratios and the user selection scheme, the resource allocation subproblem can be reformulated as
\begin{align}\label{Q24}
 	& \mathop {\max }\limits_{{\boldsymbol{B}},{\boldsymbol{P}}} \sum\limits_{i = 1}^U {{\mu _i}g_i} \\
 	\rm{s.t.}\;\; & (\ref{Q1}a) - (\ref{Q1}d), 
 \end{align}
where ${\mu _i}={{\beta _i}{w _i}\eta ({o_i})}$ is a constant independent of the resource allocation scheme. Problem (\ref{Q24}) has the same form as the problem (\ref{Q6}), and both of them are non-convex and have linear constraints. Therefore, we can also use the SCA approach to transform (\ref{Q24}) into an approximated convex problem and solve it via the dual method. The detailed solution process is omitted here.

Finally, the three subproblems are iteratively solved until a convergent solution is obtained, and the iterative algorithm is summarized in \textbf{Algorithm} \ref{CRRAUS}.
\begin{algorithm}[htbp]
	\normalsize
	\caption{CRRAUS Algorithm.}
	\begin{algorithmic}[1]
		\STATE Initialize semantic compression ratio $\boldsymbol o$, user selection ${\boldsymbol{\beta }}$, resource allocation $\boldsymbol B$ and $\boldsymbol P$.
		\REPEAT
		\STATE With fixed user selection ${\boldsymbol{\beta }}$ and resource allocation $\boldsymbol B$ and $\boldsymbol P$, optimize semantic compression ratios $\boldsymbol o$ with the enumeration method.
		\STATE With fixed semantic compression ratio $\boldsymbol o$ and resource allocation $\boldsymbol B$ and $\boldsymbol P$, optimize user selection ${\boldsymbol{\beta }}$ with branch and bound method.
		\STATE With fixed semantic compression ratios and user selection ${\boldsymbol{\beta }}$, obtain the optimal resource allocation $\boldsymbol B$ and $\boldsymbol P$ with SCA approach.
		\UNTIL {the objective value (\ref{Q1}) converges.}
	\end{algorithmic}
	\label{CRRAUS}
\end{algorithm}

\vspace{-0.2cm}
\subsection{Convergence and Complexity Analysis}
This subsection analyzes the convergence and computational complexities of CRRA and CRRAUS. 

The convergence of CRRA mainly depends on the resource allocation subproblem, while the first subproblem is solved by the one-dimension enumeration method. Thus, we focus on analyzing the convergence of Algorithm \ref{SCA}, which is illustrated by the following lemma.

\noindent\textbf{$Lemma\ 2.$} The total success probability of tasks obtained in Algorithm \ref{SCA} is monotonically non-decreasing, and the sequence (${\boldsymbol{B}^{(n)}},{\boldsymbol{P}^{(n)}}$) converges to a point fulfilling the KKT optimal conditions of the original non-convex problem (\ref{Q6}).
\begin{proof}
	Since Lemma 2 directly follows from Proposition 3 in \cite{Zappone}, the proof of Lemma 2 is omitted.
\end{proof}
The major complexity in each iteration lies in solving the semantic compression ratios subproblem and the resource allocation subproblem. With fixed resource allocation, the complexity of using the enumeration method is ${\cal O}({K^U})$ for solving (\ref{Q4}), where $K$ is the number of enumerations. With fixed compression ratios, the complexity of solving (\ref{Q8}) is ${\cal O}({U^{3.5}})$. As a result, the total complexity of CRRA is given by ${\cal O}({T_0}{K^U} + {T_0}{U^{3.5}})$, where $T_0$ is the number of iterations in CRRA.

Since the convergence analysis of CRRAUS is similar to that of CRRA, the detailed analysis is omitted. The computational complexity of CRRAUS consists of three parts. The first part is for solving the problem (\ref{Q22}) by one-dimension enumeration method, the second part is for solving 0-1 integer programming problem (\ref{Q23}) via branch and bound method, and the third part is for solving the problem (\ref{Q24}) by SCA. In step 3) of Algorithm \ref{CRRAUS}, the complexity of solving problem (\ref{Q22}) is ${\cal O}(K^U)$. The complexities of steps 4) and 5) are ${\cal O}(U^3)$ and ${\cal O}(U^{3.5})$, respectively. Therefore, the total complexity of CRRAUS is ${\cal O}((U^3+U^{3.5}+K^U)T_1)$, where $T_1$ is the number of iterations in Algorithm \ref{CRRAUS}. CRRAUS completes user selection at the expense of higher algorithm complexity than CRRA. Besides, it can be observed that the complexity of the two algorithms increases sharply with the increase 
in the number of users, which can be further optimized in future work.

\vspace{-0.1cm}
\section{Simulation Results and Analysis}
\label{sec:simulation}

\begin{table}[t]
	\normalsize
	\vspace{-0.1cm}
	\centering
	\caption{Simulation and Hyper Parameters}
	\setlength{\abovecaptionskip}{-0.5cm}
	\begin{tabular}{cc}
		\toprule
		\textbf{Simulation Parameter} & \textbf{Value} \\
		\midrule
		Initial data size, ${{d_0}}$ & 24.5 MB \\
		Delay constraint of users, $t_0$ & 1-10 ms\\
		Noise power spectral density, $N_0$ & -174 dBm/Hz\\
		Minimum bandwidth, ${{B_{\min}}}$ & 0.01 MHz\\
		Minimum transmit power, ${P_{\min}}$ & -20dBm\\
		The number of users, $U$ & 10\\
		Compression ratio, $o$ & 0-1\\
		Maximum bandwidth, $B_{\max}$ & 1 MHz-30 MHz\\
		Maximum transmit power, $P_{\max}$ & 1 mW-1 W\\
		Weights of service levels,  ${{\varepsilon}}$ & [0.2,0.4,0.6,0.8]\\
		\midrule
		\textbf{Hyper Parameter} & \textbf{Value} \\
		\midrule
		Epoch & 50 \\
		Batchsize & 32\\
		Optimizer & Adam\\
		Learning rate & 0.01\\
		Momentum & 0.9\\
		$\kappa$ & $10^{-3}$\\
		\toprule
	\end{tabular}
	\vspace{-0.cm}
	\label{tab:simulation}
\end{table}

\begin{table}[t]
\normalsize
\centering
\caption{\centering The DNN Structure for Classification}
\begin{tabular}{|c|c|c|c|}
  \hline 
  \cline{2-4}  &Layer & Output Size & Activation\\
  \hline
  \multirow{4}{*}{Transmitter} & Conv Layer & 64$\times$112$\times$112 & Relu\\ 
  \cline{2-4}  & ResNet Block & 128$\times$28$\times$28 & Relu\\
  \cline{2-4}  & ResNet Block & 512$\times$7$\times$7 & Relu\\ 
  \cline{2-4}  & Pooling Layer & 512 & None\\
  \hline
  Channel & Dense Layer & None & None\\ 
  \hline
  \multirow{3}{*}{Receiver} & Dense Layer & 256 & Relu\\ 
  \cline{2-4}  & Dense Layer & 128 & Relu\\ 
  \cline{2-4}  & Output Layer & 10 & Softmax\\
  \hline
\end{tabular}
\label{tab:network}
\end{table}

In our simulation, a circular network is considered with one edge server and $U = 10$ users. Unless specifically stated, the simulation parameters are listed in Table \ref{tab:simulation}. In the experiments, we take the image classification task as an example to illustrate. STL-10 dataset \cite{stl-Coates} is used as training and testing data, which contains images of 10 categories of objects, corresponding to 10 semantic concepts. To verify the applicability of the proposed semantic communication system to different neural networks, experiments are conducted based on the backbone of VGG\cite{VGG16} and Resnet\cite{Resnet18} networks. The network structure based on Resnet is shown in Table \ref{tab:network}. We deploy convolutional layers at the transmitter to extract a compact representation. Besides, we add a pooling layer at the end of the transmitter to reduce the communication overhead. Correspondingly, several dense layers are adopted at the receiver for further processing and outputting the task results. The hyperparameters during training are listed in Table \ref{tab:simulation}. Performing network inference under different channel signal-to-noise ratios (SNR) after the network is trained, we can obtain the parameters of the intelligent task performance model proposed in \ref{subsec:intelligent}, which are listed in Table \ref{tab:fitting}. From Table \ref{tab:fitting}, we find that the points of ${{\cal D}}$ are well approximated by the exponential function with extremely small reconstruction errors, which is quantified by root-mean-square error (RMSE).

In the following, the proposed semantic communication system with ASC (labeled as ”ASC”) is first compared with the traditional communication method (labeled as "TCM”). In the traditional communication method, the image is encoded by JPEG and transmitted to complete the task. Then, we compare the proposed CRRA and CRRAUS algorithms with three baselines: resource allocation scheme with fixed compression ratios (labeled as ”FCR”), compression ratios optimization scheme with fixed resource allocation (labeled as ”FRA”), and conventional resource allocation scheme to maximize the system sum-rate (labeled as ”MSR”).
\begin{table}[t]
	\normalsize
	\centering
	\caption{\centering Parameters of task performance model}
	\subtable[Based on the backbone of VGG]{
		\begin{tabular}{|c|c|c|c|}
			\hline
			\diagbox{${\boldsymbol{\zeta }}$}{SNR}&$-$5dB&0dB&5dB\\ 
			\hline
			$\zeta_1$&$-$9.503e$-$17&$-$2.202e$-$16&$-$2.76e$-$18\\
			\hline
			$\zeta_2$&36.77&35.94&40.33\\
			\hline
			$\zeta_3$&0.9044&0.9137&0.9205\\
			\hline
			$\zeta_4$&$-$0.01869&$-$0.02349&$-$0.02257\\
			\hline
			RMSE&0.0449&0.0488&0.0510\\
			\hline
		\end{tabular}
		\label{tab:firsttable}
	}
	\\[12pt]
	\subtable[Based on the backbone of Resnet]{        
		\begin{tabular}{|c|c|c|m{1.7cm}<{\centering}|}
			\hline
			\diagbox{${\boldsymbol{\zeta }}$}{SNR}&$-$5dB&0dB&5dB\\ 
			\hline
			$\zeta_1$&$-$6.205e$-$08&$-$2.893e$-$16&$-$8.875e$-$16\\
			\hline
			$\zeta_2$&16.45&35.68&34.54\\
			\hline
			$\zeta_3$&0.9228&0.9482&0.9458\\
			\hline
			$\zeta_4$&$-$0.06917&$-$0.04151&0.007934\\
			\hline
			RMSE&0.0272&0.0282&0.0491\\
			\hline
		\end{tabular}
		\label{tab:secondtable}
	}
	\label{tab:fitting}
\end{table}

\begin{figure}[htbp]
	\centering
	\subfigure[Based on the backbone of VGG]{
		\includegraphics[width=1\linewidth]{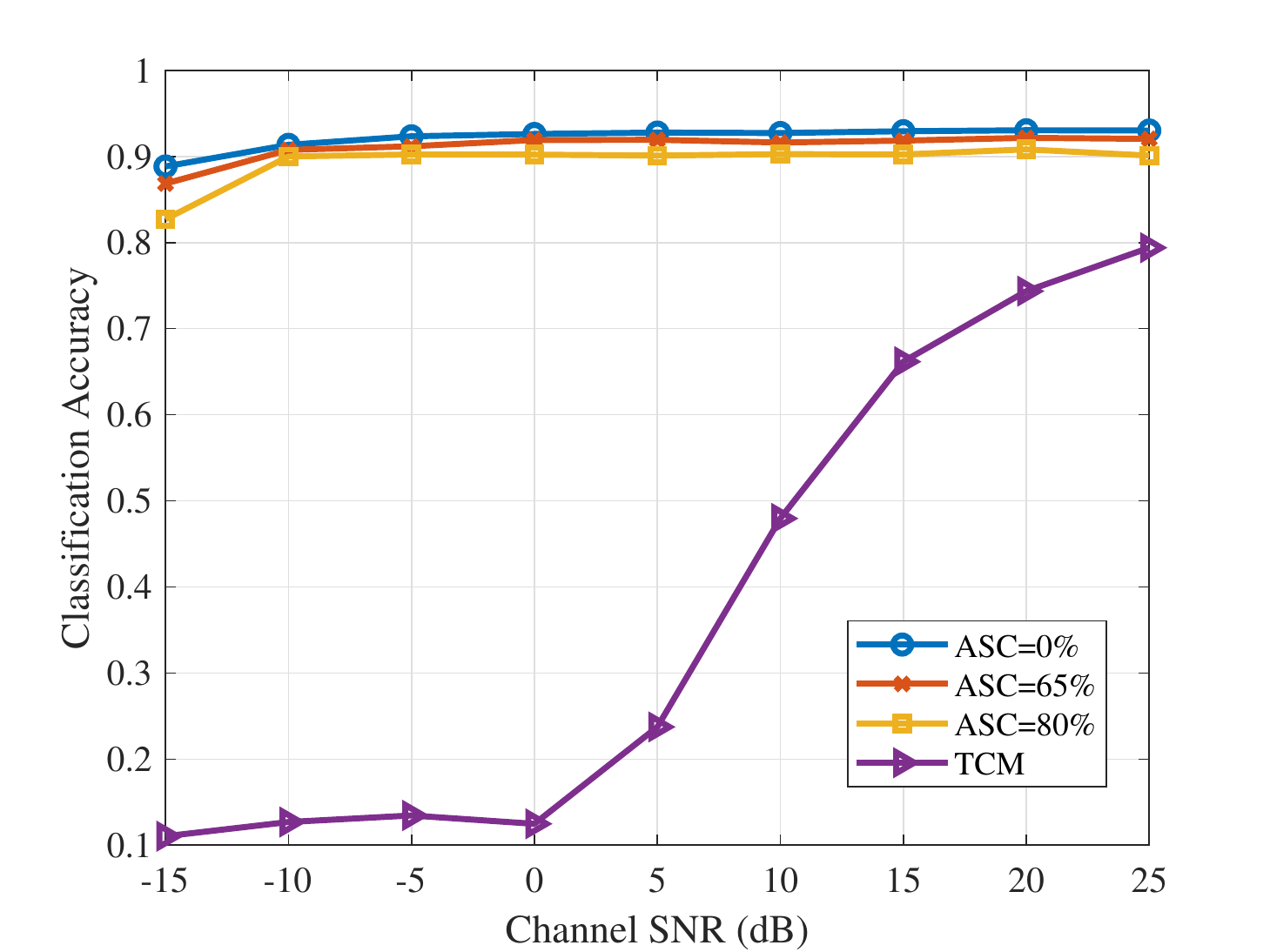}
	}
	\subfigure[Based on the backbone of Resnet]{
		\includegraphics[width=1\linewidth]{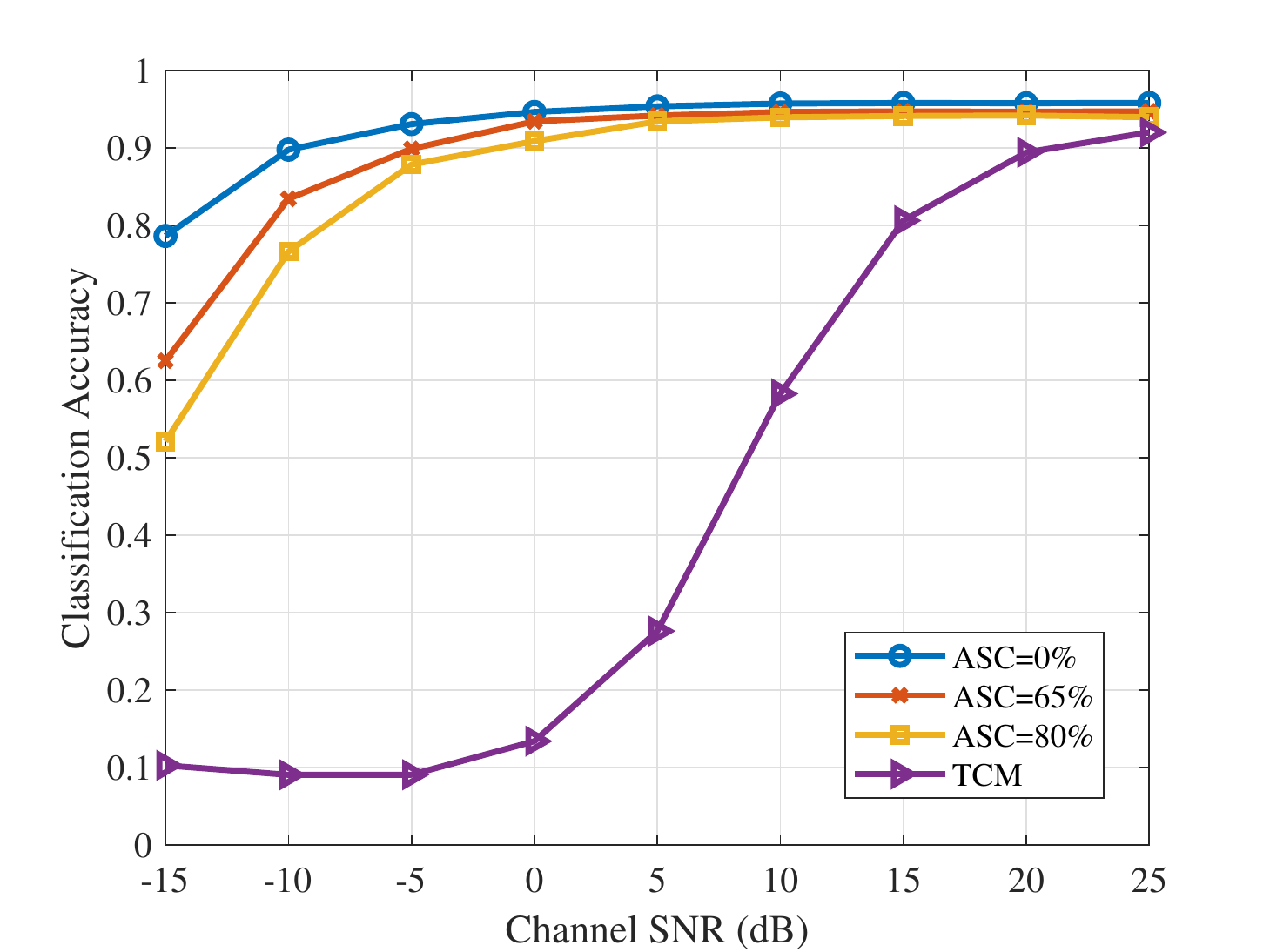}
	}
	\caption{Classification accuracy versus channel SNR under different communication systems.}
	\label{fig:system_compare}
	\vspace{-0.3cm}
\end{figure}

Fig. \ref{fig:system_compare} shows the classification accuracy versus channel SNR under different communication systems, where Fig. \ref{fig:system_compare}(a) is conducted based on the backbone of VGG and Fig. \ref{fig:system_compare}(b) is conducted based on the backbone of Resnet. "ASC=0\%", "ASC=65\%", and "ASC=80\%" refer to the user employing the proposed ASC with the semantic compression ratio equaling to 0\%, 65\%, and 80\%, respectively, where "ASC=0\%" is equivalent to the existing semantic communication system. As shown in Fig. \ref{fig:system_compare}, the performance of DL-based semantic communications is much better than that of TCM, especially in low SNR regimes. In addition, compared with "ASC=0\%", "ASC=65\%" and "ASC=80\%" will suffer a loss in classification performance. However, it can be seen from the two experimental results that when the compression ratio reaches 80\%, the loss of classification accuracy is tiny when the SNR is greater than 0 dB. This proves that the proposed ASC approach can greatly reduce the amount of transmitted data, and thus reduce the transmission delay without affecting the task performance, which is more suitable for resource-limited scenarios.

\begin{figure}[t]
	\begin{center}
		\includegraphics[width=1\linewidth]{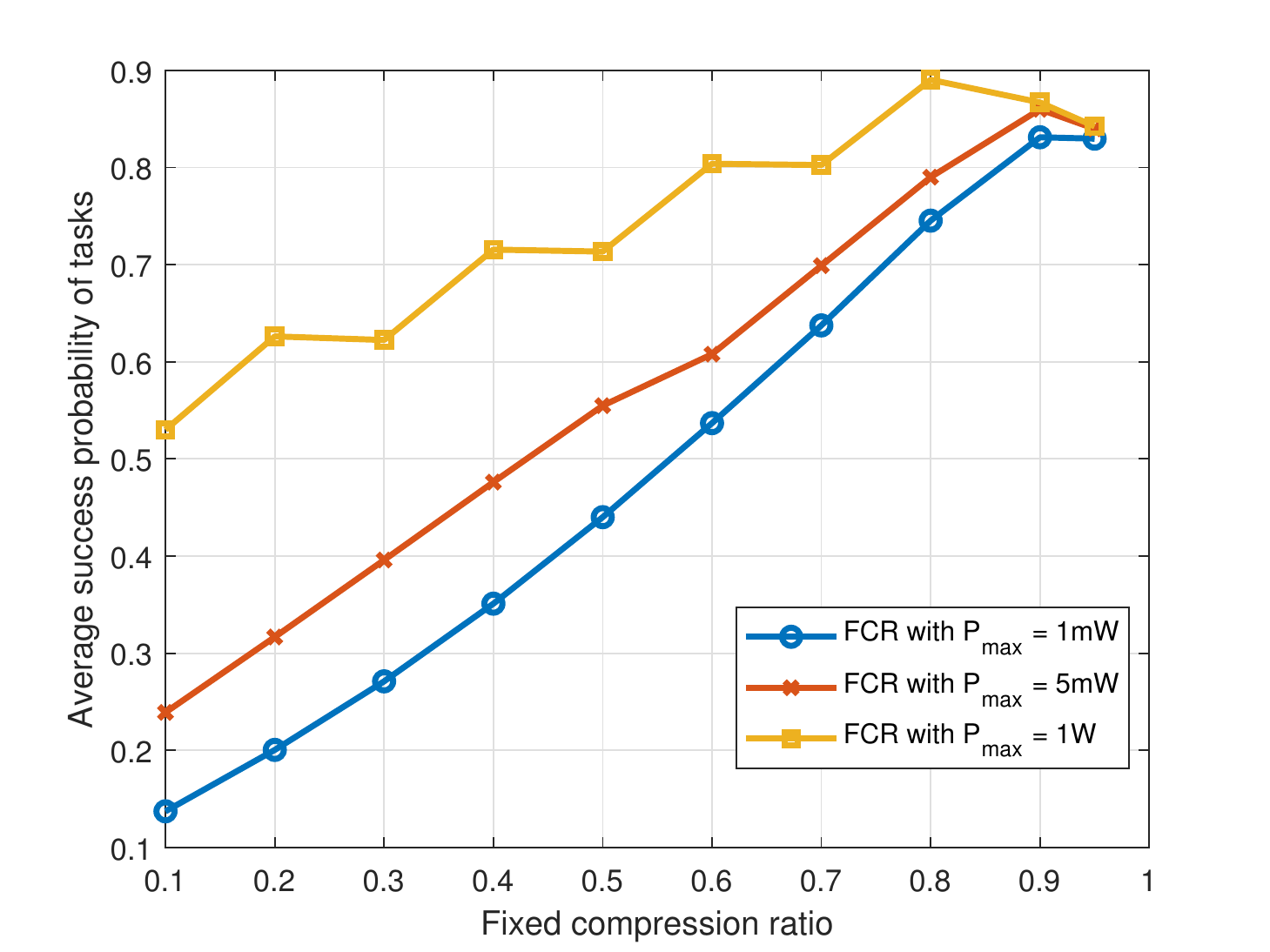}
	\end{center}
	\caption{Average success probability of tasks versus the fixed compression ratio under different maximum transmit power.}
	\label{fig:CR_power}
\end{figure}

\begin{figure}[t]
	\begin{center}
		\includegraphics[width=1\linewidth]{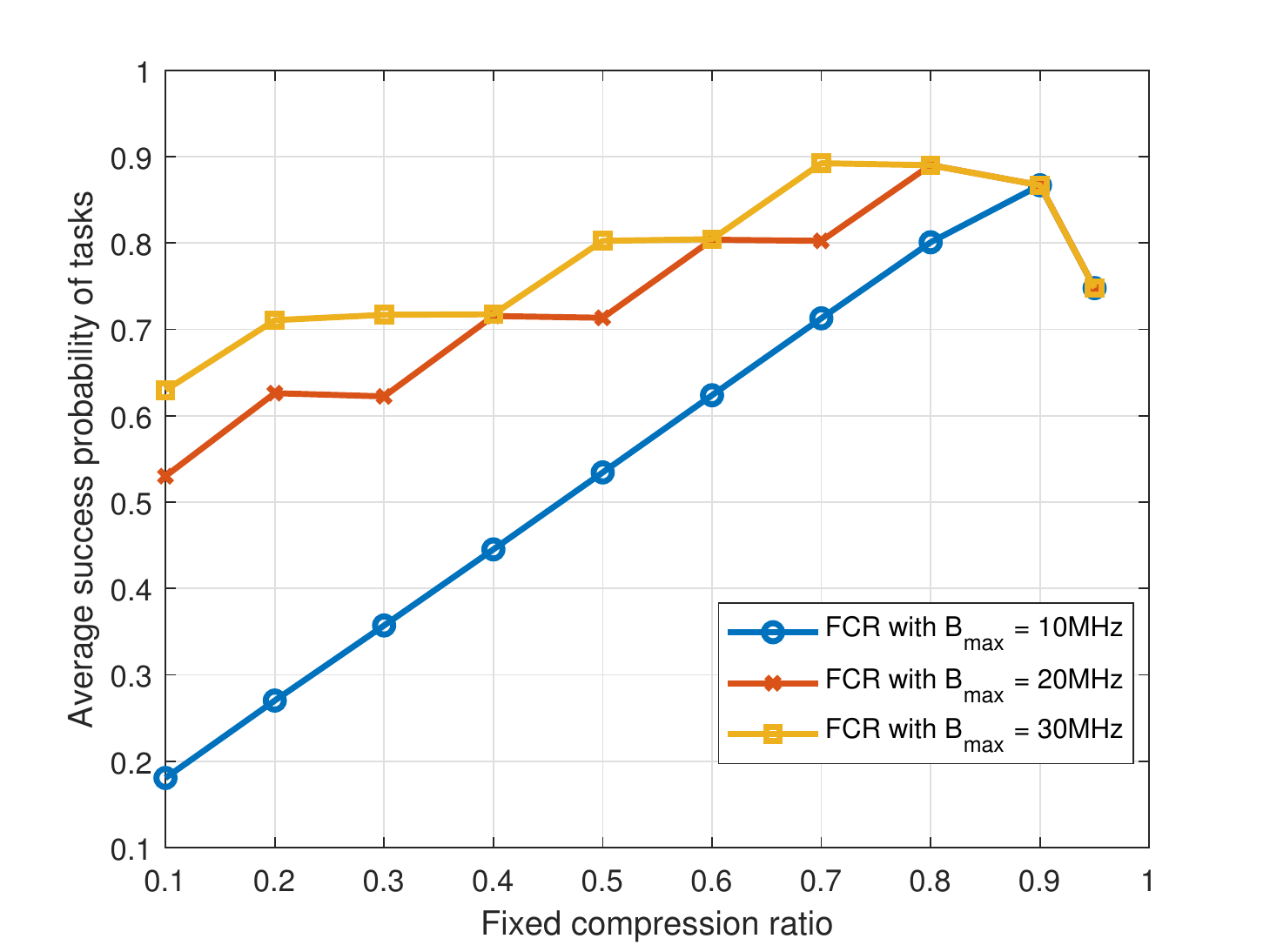}
	\end{center}
	\caption{Average success probability of tasks versus the fixed compression ratio under different maximum bandwidth.}
	\label{fig:CR_bandwidth}
	\vspace{-0.3cm}
\end{figure}
Figs. \ref{fig:CR_power} and \ref{fig:CR_bandwidth} illustrate the average success probability of tasks versus different compression ratios under different maximum transmit power and different maximum bandwidth, respectively. When the maximum transmit power and maximum bandwidth change, it can be observed that the optimal compression ratios are variable (for example, when the maximum bandwidth is 10MHz, the optimal compression ratio is 0.7, and when the maximum bandwidth is 20MHz, the optimal compression ratio is 0.8), which verifies that resources will affect the optimal compression ratios and the necessity of optimizing the compression ratios. We can also observe that the average success probability of tasks increases first and then decreases as the compression ratio increasing. This is because the choice of compression ratios is a trade-off between communication transmission and task performance, which reflects that choosing the optimal compression ratios is of great significance for semantic communications. The maximum bandwidth and the maximum transmit power of subsequent simulations are set to 20MHz and 1W, respectively, and thus the fixed compression ratio of FCR in subsequent simulations is set to 0.8 for a fair comparison.

\begin{figure}[t]
	\begin{center}
		\includegraphics[width=1\linewidth]{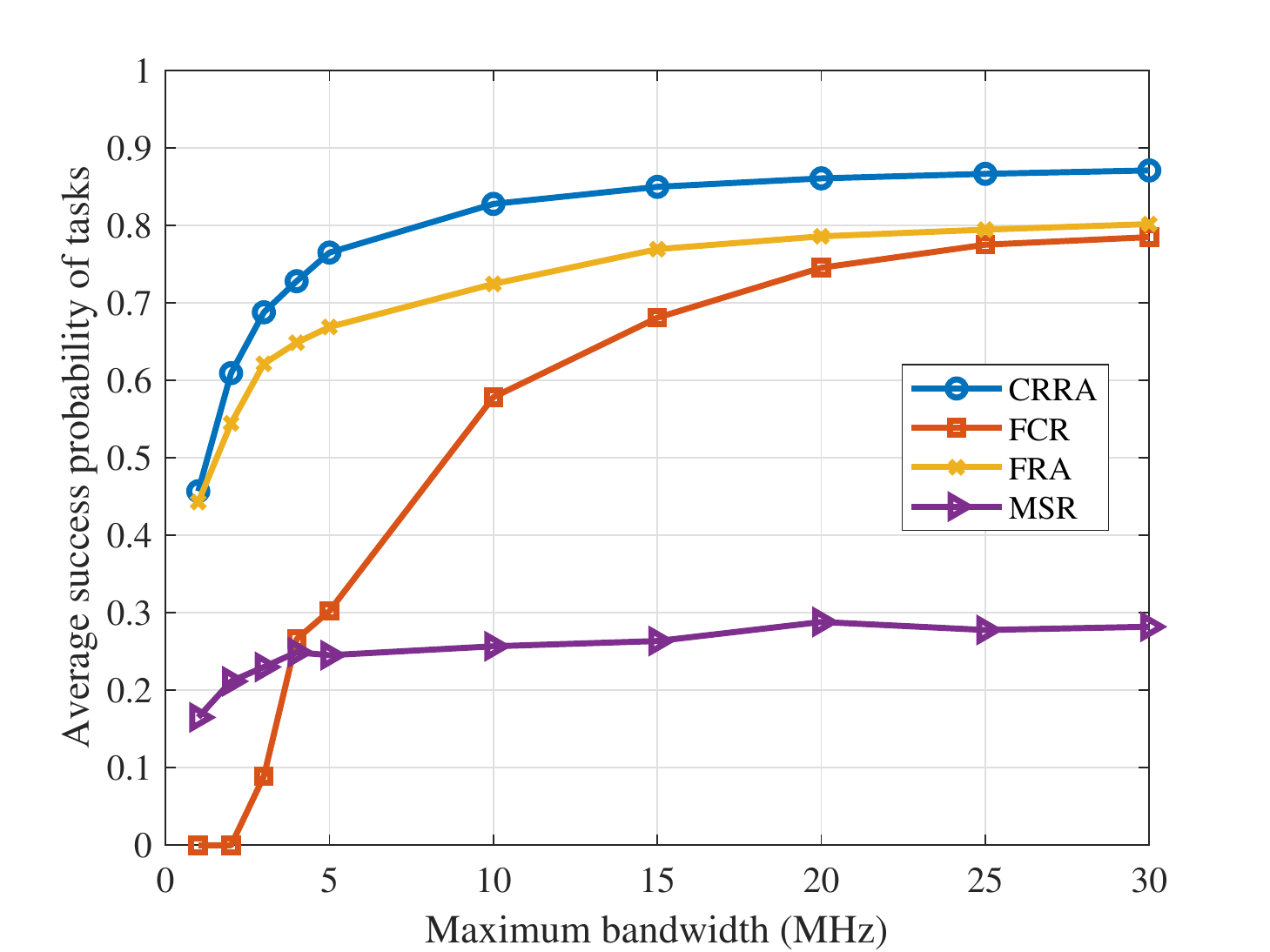}
	\end{center}
	\caption{Average success probability of tasks versus the maximum bandwidth with $P_{\rm{max}} = 1$mW.}
	\label{fig:bandwidth1}
\end{figure}

\begin{figure}[t]
	\begin{center}
		\includegraphics[width=1\linewidth]{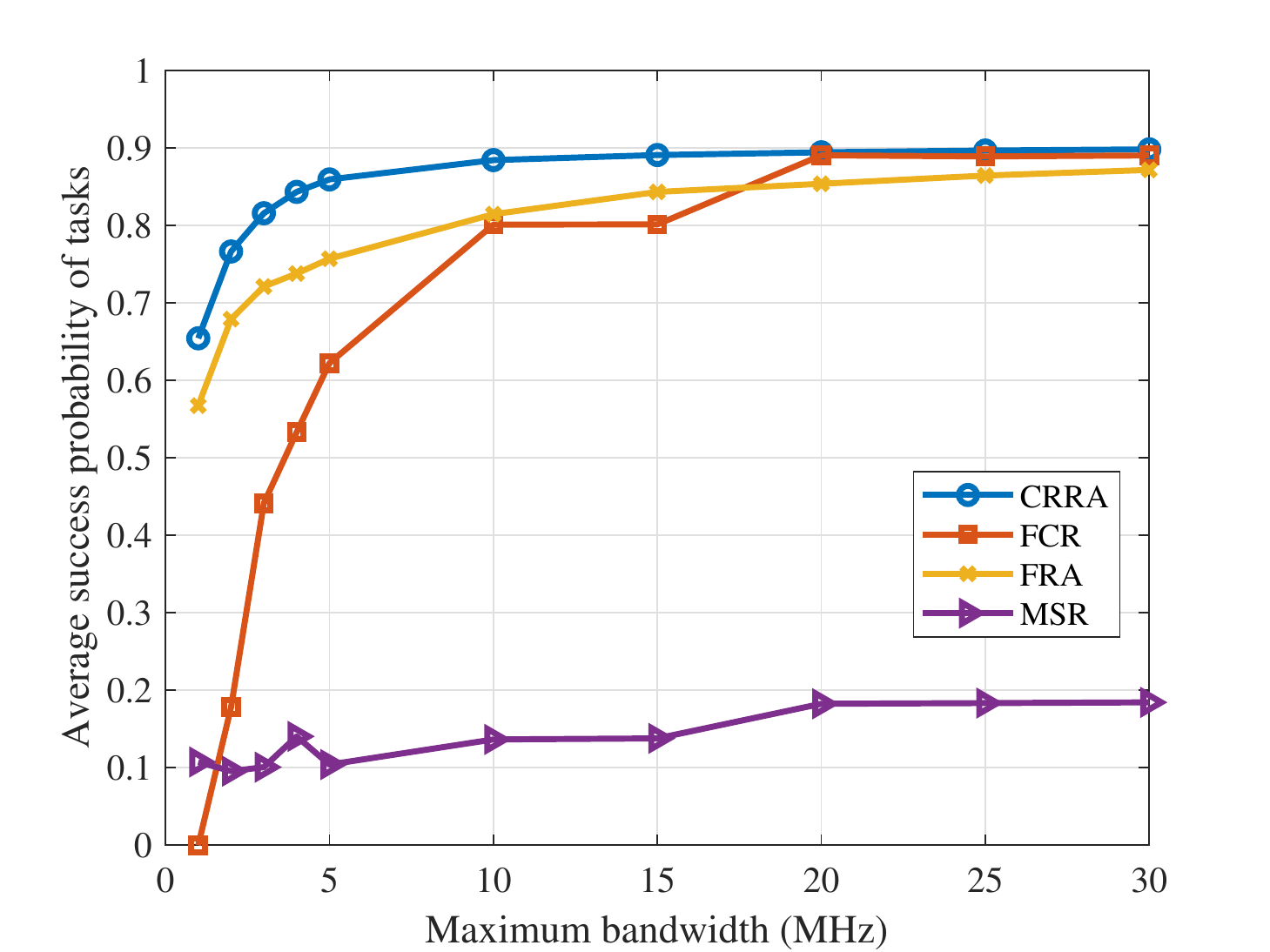}
	\end{center}
	\caption{Average success probability of tasks versus the maximum bandwidth with $P_{\rm{max}} = 1$W.}
	\label{fig:bandwidth2}
	\vspace{-0.2cm}
\end{figure}
The average success probability of tasks versus the maximum bandwidth under different maximum transmit power are shown in Figs. \ref{fig:bandwidth1} and \ref{fig:bandwidth2}. As shown in these figures, the average success probability of tasks increases with the maximum bandwidth and gradually converges to a certain threshold. This is because large bandwidth can decrease the transmission delay and tolerate a small semantic compression ratio, which consequently increases the probability of successful transmission and the average success probability of tasks. It can be observed that the average success probability of tasks of the proposed algorithm is always higher than that of others, especially in low bandwidth regions. It can be found that the conventional resource allocation scheme MSR is no longer suitable for semantic communication scenarios. This is because the conventional resource allocation scheme only optimizes the transmission rate and lacks the consideration of semantics and subsequent intelligent tasks.

\begin{figure}[t]
	\begin{center}
		\includegraphics[width=1\linewidth]{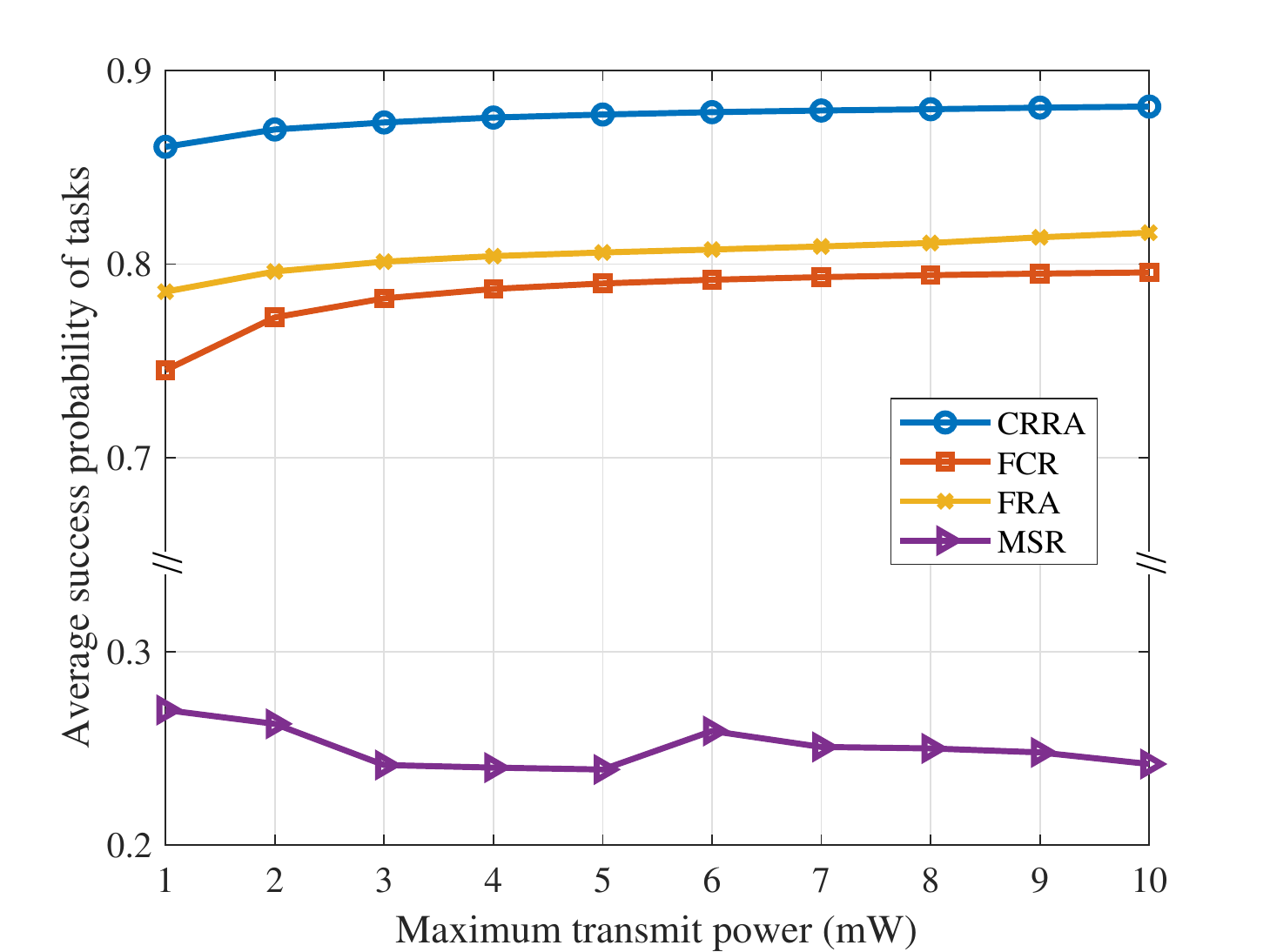}
	\end{center}
	\caption{Average success probability of tasks versus the maximum sum transmit power.}
	\label{fig:power}
\end{figure}
The average success probability of tasks versus the maximum transmit power is depicted in Fig. \ref{fig:power}. From this figure, we can observe that the proposed CRRA achieves better performance than FCR, FRA, and MSR. Fig. \ref{fig:power} demonstrates that the average success probability of tasks increases as the maximum transmit power. This is because large transmit power can increase the transmission rate, consequently increasing the amount of transmitted semantics. It can also be observed that the proposed algorithm harvests significant performance gains compared with the baselines. From Fig. \ref{fig:power}, we can further find that the MSR has little improvement in semantic performance. This is because the MSR method only focuses on technical performance, which may not transmit the semantic information required for intelligent tasks well. Besides, the proposed algorithm can perform well even in very low transmit power regions, which shows that our algorithm is very suitable for low-power scenarios.

\begin{figure}[t]
	\begin{center}
		\includegraphics[width=1\linewidth]{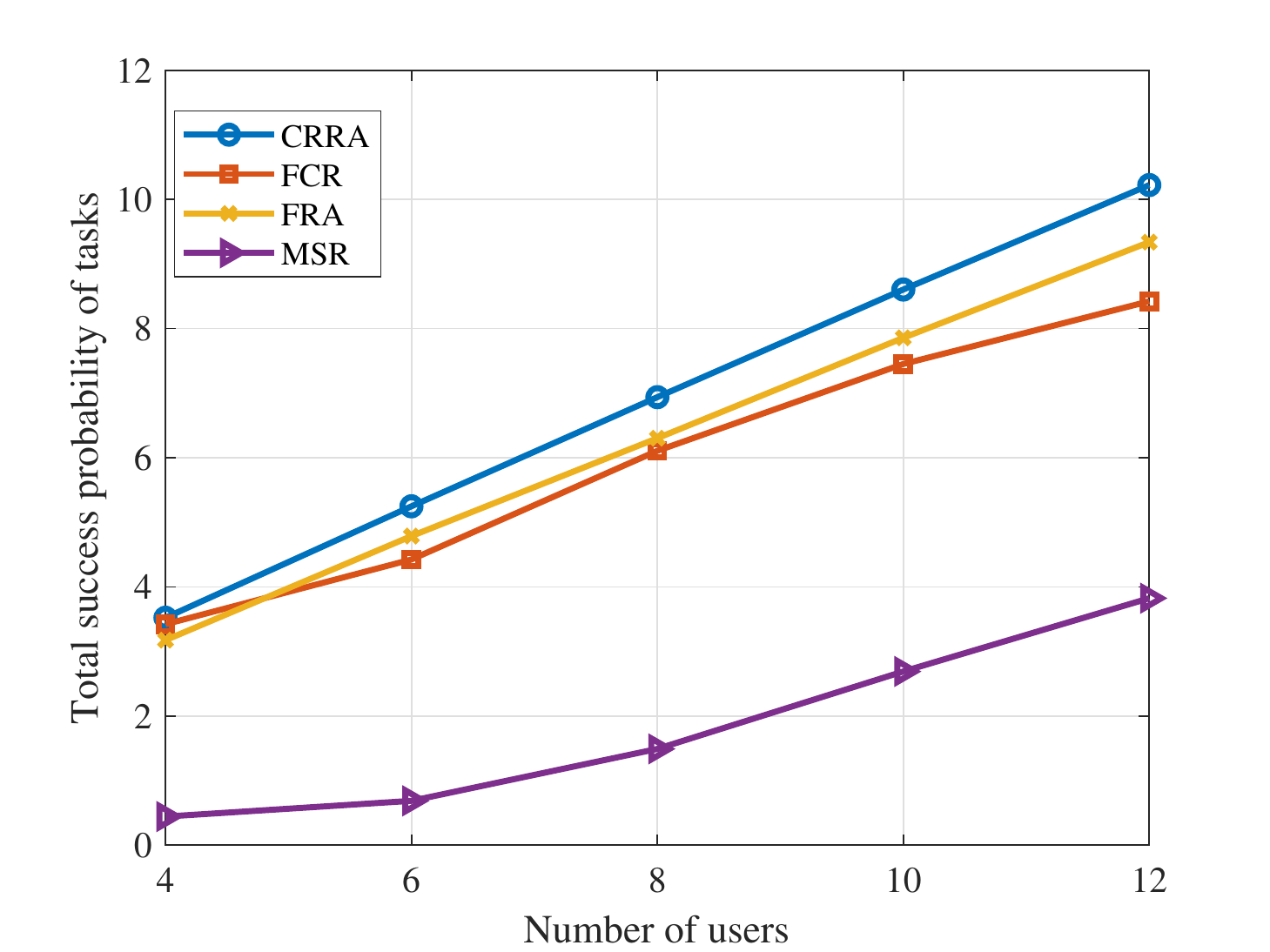}
	\end{center}
	\caption{Total success probability of tasks versus number of users.}
	\label{fig:user}
	\vspace{-0.3cm}
\end{figure}
The total success probability of tasks versus the number of users is given in Fig. \ref{fig:user}. Clearly, the proposed CRRA is always better than FCR, FRA, and MSR, especially when the number of users is large. This is because CRRA can effectively determine the compression ratios and the resource allocation scheme to meet the delay constraint, while FCR and FRA only take one of them into consideration. MSR has the worst performance because only maximizing the sum rate cannot guarantee an accurate understanding of semantic information. When the number of users is large, the multi-user gain is more apparent by the proposed CRRA compared to conventional FCR and FRA. This is because the resources are relatively tight when there are a large number of users, and CRRA can make full use of resources and find the optimal trade-off between compression and transmission. CRRA achieves better performance than FCR and FRA at the cost of additional computational complexity.

\begin{figure}[t]
	\begin{center}
		\includegraphics[width=1\linewidth]{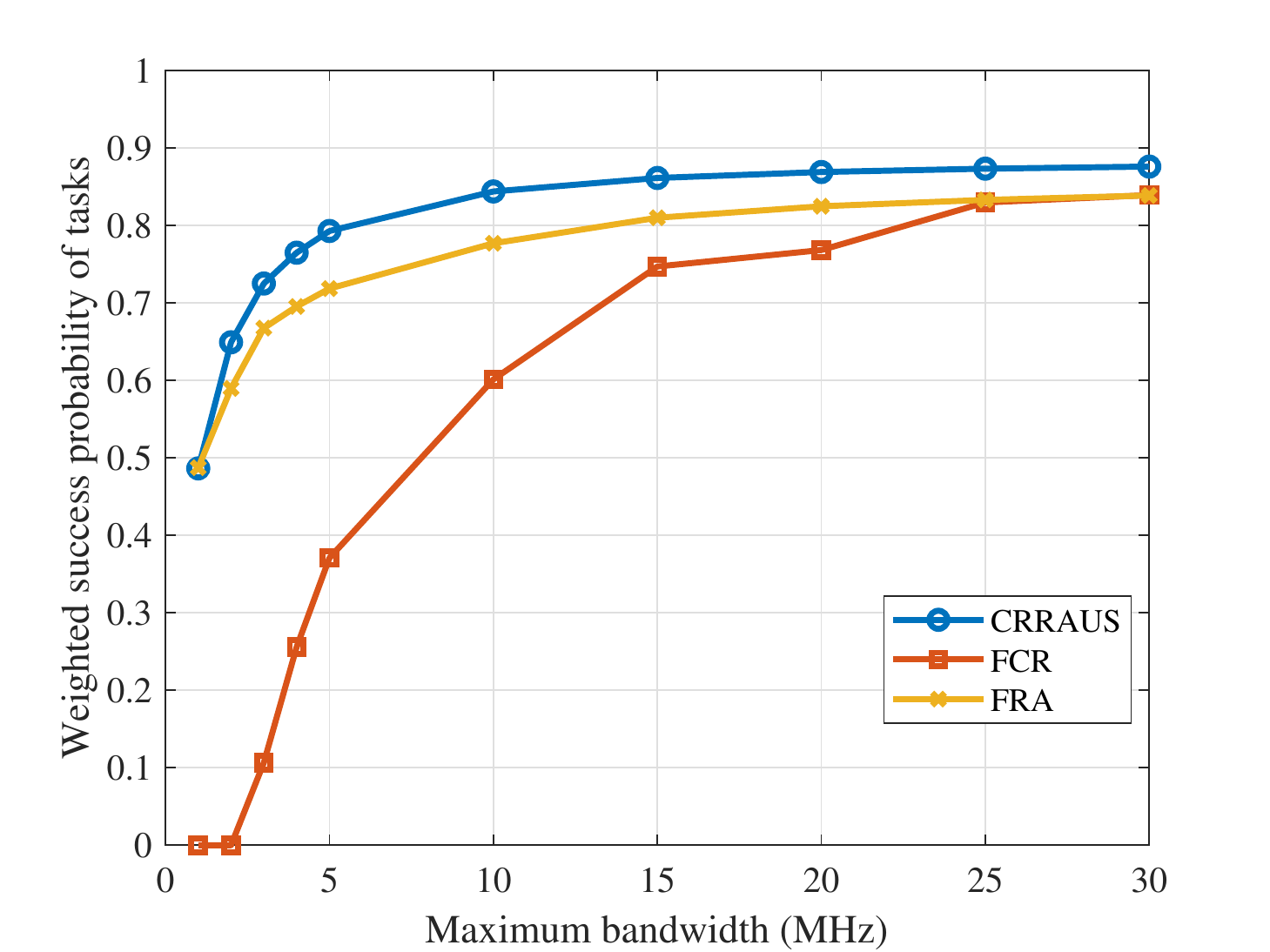}
	\end{center}
	\caption{Weighted success probability of tasks versus the maximum bandwidth.}
	\label{fig:wCRRA_bandwidth}
\end{figure}
Fig. \ref{fig:wCRRA_bandwidth} shows how the weighted success probability of tasks changes as the maximum bandwidth. From fig. \ref{fig:wCRRA_bandwidth}, we can see that the weighted success probability of tasks rises as maximum bandwidth increases. This is because the larger the bandwidth, the more resources users are allocated, and the better the performance of semantic communications will be. The proposed CRRAUS outperforms FRA and FCR in terms of the weighted success probability of tasks, particularly for cases with a small bandwidth. This is because CRRAUS can simultaneously optimize the resource allocation, compression ratios, and user selection, while comparison schemes can only optimize one of them separately.

\begin{figure}[t]
	\begin{center}
		\includegraphics[width=1\linewidth]{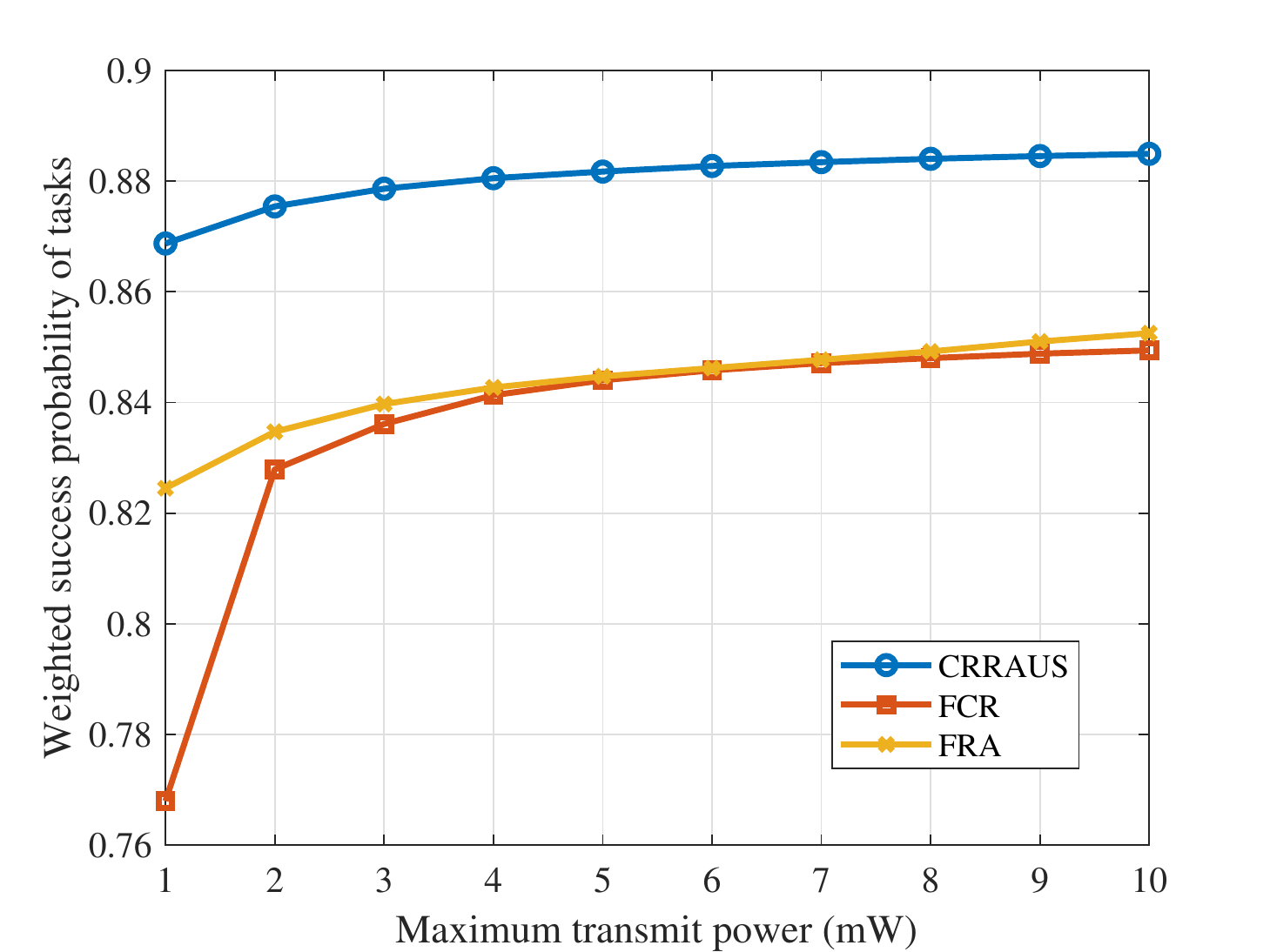}
	\end{center}
	\caption{Weighted success probability of tasks versus maximum sum transmit power.}
	\label{fig:wCRRA_power}
\end{figure}

The weighted success probability of tasks versus the maximum transmit power is given in Fig. \ref{fig:wCRRA_power}. From this figure, the weighted success probability of tasks increases for all schemes as the maximum transmit power varies. This is because high transmit power can increase the amount of transmitted semantics, which improves the performance of semantic communication. It is observed that the proposed CRRAUS achieves the best performance under different maximum transmit power. This is because users are adaptively selected based on resources and service levels in CRRAUS algorithm, while the comparison schemes select all users regardless of service level and without joint optimization, which verifies the superiority of joint optimization of compression ratios, resource allocation, and user selection.

\vspace{-0.cm}
\section{Conclusion}
\label{sec:conclusion}
In this paper, we have investigated performance optimization for task-oriented multi-user semantic communications. Specifically, we have first developed a task-oriented multi-user semantic communication system, in which an ASC approach is proposed to compress semantics to reduce the communication burden adaptively. Then, we have formulated a resource allocation and compression ratios optimization problem under bandwidth and power constraints to maximize the success probability of tasks, which is defined to measure the performance of semantic communications. For scenarios where users have the same service levels, we have proposed a CRRA algorithm to optimize resource allocation and compression ratios, where the nonconvex problem is decomposed into two subproblems and solved iteratively. Furthermore, considering that users have various service levels, a CRRAUS algorithm has been proposed, in which users are adaptively selected based on the branch and bound method. Simulation results have shown that the proposed ASC approach can significantly reduce the size of transmitted data, and both CRRA and CRRAUS algorithms achieve higher success probability of tasks than the benchmarks, especially when communication resources are tight. Compared with CRRA, the CRRAUS is more suitable for scenarios with significant differences in service levels at the expense of higher complexity. Future extensions of this work will further explore a unified semantic importance measurement method and reduce the computational complexity of the algorithms.

\begin{appendices}
\section{$Proof\ of\ Lemma\ 1$}
\label{proof_lemma1}
\begin{proof}
	Based on (\ref{R}) and (\ref{t}), we have
	\begin{align}\label{Q2}
		P({t_i} \le {t_0}) = &{\rm{P}}\left( {\frac{{\left( {1 - {o_i}} \right){d_0}}}{{{B_i}{{\log }_2}\left( {1 + \frac{{{h_i}{P_i}}}{{{N_0}{B_i}}}} \right)}} \le {t_0}} \right)\nonumber\\
		= &{\rm{P}}\left( {\frac{{{2^{{a_i}(1 - {o_i})}} - 1}}{{{b_i}}} \le {h_i}} \right)\nonumber\\
		= &2Q\left( {\frac{{{2^{{a_i}(1 - {o_i})}} - 1}}{{{b_i}\delta }}} \right),
	\end{align}
	where the last equality follows from ${h_i} \sim N(0,{\delta ^2})$.	

	This completes the proof of Lemma 1.
\end{proof}

\vspace{-0.2cm}
\end{appendices}
\bibliographystyle{IEEEbib}
\nocite{*}\bibliography{stimreference}

\end{document}